%% file: main-ArXiv.tex
\newif\iflnbi
\lnbifalse

\iflnbi
  \newcommand{\mysubparagraph}[1]{\medskip\noindent{\sf\bf #1}}
  \documentclass[runningheads]{llncs}

  \usepackage{etex}
  \usepackage{hyperref}
   
  \usepackage{amsthm,amssymb,amsmath}
  \usepackage[nameinlink]{cleveref}

  \newenvironment{claimproof}[1][\proofname]{%
    \begin{proof}[#1]%
  }{%
    \end{proof}%
  }
  \newtheorem{observation}{Observation}

\else
  \newcommand{\mysubparagraph}{\subparagraph*}
  \documentclass[a4paper,USenglish,cleveref,autoref,thm-restate,numberwithinsect,nameinlink]{lipics-v2021}

  \usepackage{etex}
  \usepackage{amsthm,amssymb,amsmath}
\fi

\usepackage[utf8]{inputenc}
\usepackage[T1]{fontenc}
\usepackage[disable]{todonotes}

\usepackage{soul}

\usepackage[shortlabels, inline]{enumitem}
\setlist[enumerate]{topsep=0pt,itemsep=-1ex,partopsep=1ex,parsep=1ex}
\setlist[itemize]{topsep=0pt,itemsep=-1ex,partopsep=1ex,parsep=1ex}
\setlist[description]{topsep=0pt,itemsep=-1ex,partopsep=1ex,parsep=1ex}

\usepackage{graphicx}
\usepackage{microtype}
\usepackage[vlined,ruled,linesnumbered]{algorithm2e}
\usepackage{lmodern}
\usepackage{pgf}
\usepackage[sort,comma,square,numbers]{natbib}
\usepackage{xspace}
\usepackage{textcomp}
\usepackage{mathtools}
\usepackage{tabularx}
\usepackage{nicefrac}
\usepackage{booktabs}

\usepackage{subcaption}
\hypersetup{
	colorlinks,
	linkcolor={red!50!black},
	citecolor={blue!90!black!80},
	urlcolor={blue!80!black}
}

\usepackage{tikz}
\usetikzlibrary{arrows.meta,shapes,positioning,calc,decorations.pathreplacing,decorations.markings,decorations.pathmorphing,patterns}
\pgfdeclarelayer{background2}
\pgfdeclarelayer{background}
\pgfdeclarelayer{foreground}
\pgfsetlayers{background2,background,main,foreground}

\usepackage[edges]{forest}

\tikzstyle{bold}=[draw, line width=2pt]
\tikzstyle{optional}=[dashed]
\tikzstyle{path}=[decorate, decoration={snake, pre length=3mm, amplitude=.6mm}]

\tikzstyle{small}=[inner sep=2pt]
\tikzstyle{tiny}=[inner sep=1.7pt]
\tikzstyle{textnode}=[inner sep=0pt]

\tikzstyle{triangle}=[draw, regular polygon, regular polygon sides=3]
\tikzstyle{vertex}=[circle, draw, fill=white]
\tikzstyle{reti}=[vertex, fill=black]
\tikzstyle{leaf}=[vertex, rectangle]
\tikzstyle{leaf2}=[vertex, regular polygon, regular polygon sides=3]

\tikzstyle{smallvertex}=[vertex, small]
\tikzstyle{smallleaf}=[leaf, inner sep=3.3pt]
\tikzstyle{smallleaf2}=[leaf2, inner sep=1.7pt]
\tikzstyle{smalltriangle}=[triangle, inner sep=1.5pt]
\tikzstyle{smallreti}=[reti, small]

\tikzstyle{match}=[edge,line width=3pt]
\tikzstyle{edge}=[draw,-]
\tikzstyle{arc}=[draw,arrows={-Latex[length=6pt]}]
\tikzstyle{boldarc}=[draw, bold, arrows={-Latex[length=10pt]}]
\tikzstyle{revarc}=[draw, arrows={Latex[length=6pt]-}]
\tikzstyle{boldrevarc}=[draw, bold, arrows={Latex[length=10pt]-}]

\newcommand{\nextnode}[5][vertex]{\node[small#1] (#2) at ($(#3)+(#4)$) {} edge[#5] (#3);}
\newcommand{\nextnodelab}[7][vertex]{\node[small#1] (#2) at ($(#3)+(#4)$) {}; \draw[#5] (#2) -- (#3) node[pos=#6] {#7};}

\newcommand{\PROB}[1]{\textsc{#1}\xspace}
\newcommand{\MaxNPD}{\PROB{Max-Network-PD}}
\newcommand{\SubProd}{\PROB{Subset Product}}
\newcommand{\PenSum}{\PROB{Penalty Sum}}
\newcommand{\XTClong}{\PROB{Exact Cover by 3-Sets}}
\newcommand{\XTC}{\PROB{X3C}}
\newcommand{\ucNAP}{\PROB{unit-cost-NAP}}

\newcommand{\NP}{{\normalfont\texttt{NP}}\xspace}

\newcommand{\floorvar}[2]{\lfloor{#2}\rfloor_{#1}}
\newcommand{\floorH}[1]{\floorvar{H}{#1}}
\newcommand{\ceilvar}[2]{\lceil{#2}\rceil_{#1}}
\newcommand{\ceilH}[1]{\ceilvar{H}{#1}}
\newcommand{\w}{\ensuremath{\omega}}
\newcommand{\Oh}{\ensuremath{\mathcal{O}}}
\newcommand{\Instance}{{\ensuremath{\mathcal{I}}}\xspace}
\newcommand{\yes}{{\normalfont\texttt{yes}}\xspace}

\newcommand{\off}{\ensuremath{\operatorname{off}}}

\newcommand{\NetPD}{\NetPDsub{\Net}}
\newcommand{\NetPDsub}[1]{\ensuremath{\text{Network-PD}_{#1}}}
\newcommand{\PD}{\PDsub{\Net}}
\newcommand{\PDsub}[1]{\ensuremath{\text{PD}_{#1}}}
\newcommand{\Net}{\ensuremath{\mathcal{N}}\xspace}
\newcommand{\Tree}{\ensuremath{\mathcal{T}}\xspace}
\newcommand{\Inst}{\ensuremath{\mathcal{I}}\xspace}
\newcommand{\N}{\ensuremath{\mathbb{N}}\xspace}

\newcommand{\gam}[2][p]{\ensuremath{\gamma^{#1}_{#2}}}

\newtheorem{rrule}{Reduction Rule}
\newtheorem{brule}{Branching Rule}
\newtheorem{construction}{Construction}

\title{Phylogenetic Network Diversity Parameterized by Reticulation Number and Beyond} 

\titlerunning{Maximizing Network Phylogenetic Diversity}

\iflnbi
	\author{
		Leo van Iersel\inst{1,}\orcidID{0000-0001-7142-4706}\thanks{
			Partially funded by the Dutch Research Council (NWO) grant OCENW.KLEIN.125 and OCENW.M.21.306.
		} \and
        Mark~Jones\inst{1,}\orcidID{0000-0002-4091-7089}${}^\star$  \and
		Jannik~Schestag\inst{1,}\orcidID{0000-0001-7767-2970}\thanks{
			Partially funded by the Dutch Research Council (NWO), project “Optimization for and with Machine Learning (OPTIMAL)” OCENW.GROOT.2019.015.
		} \and
		Celine~Scornavacca\inst{2,}\orcidID{0009-0004-0179-9771}\thanks{
			Partially funded by French Agence Nationale de la Recherche through the CoCoAlSeq Project (ANR-19-CE45-0012).
		} \and
		Mathias~Weller\inst{3,}\orcidID{0000-0002-9653-3690}
	}
	\institute{
	TU Delft, The Netherlands
	  \email{\{l.j.j.vanIersel,j.t.schestag\}@tudelft.nl}\and
	ISEM, Université de Montpellier, CNRS, IRD, EPHE, France
	  \email{celine.scornavacca@umontpellier.fr} \and
	LIGM, CNRS, Université Gustave Eiffel, France
	  \email{mathias.weller@u-pem.fr}
	}
\else
  \author{Leo van Iersel}
    {TU Delft, The Netherlands \and \url{https://leovaniersel.wordpress.com/}}
    {l.j.j.vanIersel@tudelft.nl}
    {https://orcid.org/0000-0001-7142-4706}
    {Partially funded by the Dutch Organisation for Scientific Research (NWO) grant OCENW.KLEIN.125 and OCENW.M.21.306.}
  \author{Mark Jones}
    {TU Delft, The Netherlands \and \url{https://www.thenetworkcenter.nl/People/Postdocs/person/83/Dr-Mark-Jones}}
    {m.e.l.jones@tudelft.nl}
    {https://orcid.org/0000-0002-4091-7089}
    {Partially funded by the Dutch Organisation for Scientific Research (NWO) grant OCENW.KLEIN.125 and OCENW.M.21.306.}
  \author{Jannik Schestag}
    {TU Delft, The Netherlands}
    {j.t.schestag@tudelft.nl}
    {https://orcid.org/0000-0001-7767-2970}
    {Partially funded by the Dutch Research Council (NWO), project “Optimization for and with Machine Learning (OPTIMAL)” OCENW.GROOT.2019.015.}
  \author{Celine Scornavacca}
    {ISEM, Université de Montpellier, CNRS, IRD, EPHE, Montpellier, France \and \url{https://sites.google.com/view/celinescornavacca}}
    {celine.scornavacca@umontpellier.fr}
    {https://orcid.org/0009-0004-0179-9771}
    {Partially funded by French Agence Nationale de la Recherche through the CoCoAlSeq Project (ANR-19-CE45-0012).}
  \author{Mathias Weller}
    {LIGM, CNRS, Univsersité Gustave Eiffel}
    {mathias.weller@cnrs.fr}
    {https://orcid.org/0000-0002-9653-3690}{}

  \Copyright{Leo van Iersel, Mark Jones, Jannik Schestag, Celine Scornavacca and Mathias Weller}

  \date{June 2023 --- \today}


\fi

\authorrunning{van Iersel, Jones, Schestag, Scornavacca, and Weller}

\nolinenumbers

\hideLIPIcs

\begin{document}
\maketitle

\begin{abstract}
	Network Phylogenetic Diversity (Network-PD) is a measure for the diversity of a set of species based on a rooted phylogenetic network (with branch lengths and inheritance probabilities on the reticulation edges) describing the evolution of those species.
	We consider the \MaxNPD problem: Given such a network, find~$k$ species with maximum Network-PD score. We show that this problem is fixed-parameter tractable (FPT) for binary networks, by describing an optimal algorithm running in~$\Oh(2^r \log(k)(n+r))$~time, with~$n$ the total number of species in the network and~$r$ its reticulation number.
	Furthermore, we show that \MaxNPD is NP-hard for level-1 networks, proving that, unless P=NP, the FPT approach cannot be extended by using the level as parameter instead of the reticulation number.
\end{abstract}

\section{Introduction}
As human activities drive a sixth mass extinction~\citep{10.1093/biosci/bix125},
and in the absence of a serious political response to this crisis~\cite{doi:10.1177/00368504231201372},
studying \emph{phylogenetic diversity} (PD) is timely.\todo{JS: We are rushing directly into PD. Shouldn't we first say that we want to preserve species and therefore a study of biodiversity is needed. Then, we can talk about PD.}

Indeed, when experiencing a widespread and rapid decline in Earth's biodiversity,
one could wonder where to put our efforts in order to preserve a maximum amount of \emph{biodiversity},
given some temporal and economic constraints~\cite{weitzman1998noah}.
The concept of PD is an attempt at answering this question.
The concept has been introduced three decades ago in an impactful paper by Daniel Faith~\citep{faith1992conservation}.
The underlying idea is simple: if we want to preserve as much biodiversity as possible within a group~$X$ of species and we can rescue at most~$k$ species,
then we should focus our effort on a size-$k$ subset~$S\subseteq X$ of species that showcase, overall, a wide range of features, that is,
the distinct traits and qualities covered by the species of~$S$ are maximum among all such subsets.
This \emph{feature diversity} (FD) of~$S$ is often approximated using the PD of~$S$, which is  in turn defined as follows:
Given a tree~$T$ representing the evolution of the species in~$X$,
the PD of $S$ (in $T$) is the sum of the branch lengths of 
the subtree connecting the root and the species in~$S$. 
(Note that approximating FD with PD may not always be appropriate, see~\citep{wicke2021formal}.)

PD has been extensively used in the context of tree-like evolution, and, given a tree $T$ and an integer $k$, an optimal solution with~$k$ species can be found with a greedy algorithm~\citep{steel2005phylogenetic,pardi2005species}.

However, when the evolution of the species under interest is also shaped by reticulate events
such as hybrid speciation, lateral gene transfer, or recombination,
then the picture is no longer as rosy.
In the case of reticulate events, a single species may inherit genetic material and, thus, features from multiple direct ancestors, and
its evolution should be represented by a phylogenetic network~\citep{huson2010phylogenetic} rather than a tree.
Several ways of extending the notion of PD for networks have been proposed~\citep{WickeFischer2018,bordewich2022complexity,MAPPD},
one of which is called Network-PD.
The optimization problem linked to Network-PD, i.e.\ computing the maximum \NetPD~score
over all subsets of species of size at most~$k$ for a given phylogenetic network $\Net$, is named \MaxNPD.
\citet{bordewich2022complexity} proved that \MaxNPD is NP-hard
 and
 cannot be approximated in polynomial time with an approximation ratio better than $1 -\frac{1}{e}$ unless \mbox{P = NP}; furthermore, it remains NP-hard even for the restricted class of phylogenetic networks called ``normal'' networks.

\looseness=-1
The contribution of this paper is twofold.
First, we present an algorithm for \MaxNPD parameterized by the reticulation number of the input network.
Herein, we leverage the greedy algorithm for PD on trees~\cite{steel2005phylogenetic,pardi2005species}
to efficiently process the subtree below a reticulate event.
Surprisingly, we show that this algorithm cannot be generalized to use the ``level'' as parameter unless \mbox{P = NP}. The level of a network is a measure of its treelikeness, formally defined in the next section, which can be smaller than the reticulation number.
More precisely, we prove that \MaxNPD is NP-hard even on level-1 networks (which are networks without overlapping cycles),
thereby answering an open question of \citet{bordewich2022complexity}.

\section{Preliminaries}

For a positive integer $n$, denote $[n]:= \{1,\dots, n\}$.
Let $(0,1):= \{x \in \mathbb{R}: 0 < x < 1\}$ and $[0,1]:= (0,1) \cup \{0,1\}$.
Let $\mathbb{R}_{>0}: = \{x \in \mathbb{R}: x > 0\}$ 
and $\mathbb{R}_{\ge 0}: = \mathbb{R}_{> 0} \cup \{0\}$.
For a set~$Z$ and an integer~$k$ with $k \leq |Z|$, by $\binom{Z}{k}$ we denote the set of all subsets of $Z$ with exactly $k$~elements.
In this paper, we make use of both natural and binary logarithms. We write $\ln x$, and $\log_2 x$, to denote the logarithm of $x$ to the base $e$ and 2, respectively.

\mysubparagraph{Phylogenetics.}
Consider the example network of \Cref{fig:example-network}.
Given a set of taxa $X$, a \emph{phylogenetic network on $X$} or \emph{$X$-network} is a directed acyclic graph $\Net = (V,E)$ in which the \emph{leaves}, vertices of indegree~$1$ and outdegree $0$, are bijectively labeled with elements from $X$, and in which the \emph{root} is the single vertex of indegree $0$ and outdegree~$2$, and in which all other vertices either are \emph{tree vertices} and have indegree~$1$ and outdegree at least~$2$ or are \emph{reticulations} and have indegree at least~$2$ and outdegree~$1$.
Edges incoming at reticulations are \emph{reticulation edges}.
When~$X$ is clear from context, we refer to an $X$-network simply as a \emph{network} or \emph{phylogenetic network}.
A \emph{phylogenetic tree $\Tree = (V,E)$ on $X$}  or \emph{$X$-tree} is an $X$-network with no reticulations.
A network is \emph{binary} if the maximum indegree and outdegree of any vertex is $2$.

\input{example_fishes.tex}


The \emph{reticulation number} of a network~$\Net$ is the sum of the indegrees of all reticulations minus the number of reticulations.
If $\Net$ is binary, then the reticulation number is exactly the number of reticulations.
The \emph{level} of $\Net$ is the maximum reticulation number
among subgraphs with no cut-arcs (arcs whose removal disconnects the network).

For each edge $e = uv$ we say that $u$ is a \emph{parent} of~$v$ and $v$ is a \emph{child} of $u$.
For vertices $u,v \in V$, we say $u$ is an \emph{ancestor of $v$} and $v$ is a \emph{descendant of $u$} if there is a directed path from $u$ to $v$ in $\Net$. If in addition $u \neq v$, we say $u$ is a \emph{strict} ancestor of $v$ and~$v$ a \emph{strict} descendant of $u$.
The set of \emph{offspring} of $e$, denoted $\off(e)$, is the set of all $x \in X$ which are descendants of $v$.
%
Throughout this paper, we use the terms taxon/taxa, species, and leaf/leaves interchangeably.

\mysubparagraph{Diversity.}
We assume that each edge $e$ in a network $\Net = (V,E)$ has an associated positive integer \emph{weight} $\omega(e)$.\todo{MW: why 'positive'? weight-0 reticulation edges make sense biologically. Can we get rid of this assumption?}
These weights are used to represent some measure of difference between two species. 
%
Given an $X$-tree $\Tree$ and a weight function~$\w:E\to\N$, the \emph{phylogenetic diversity} $\PDsub{\Tree}(Z)$ of any subset $Z \subseteq X$ is given by 
%
  $\PDsub{\Tree}(Z) := \sum_{e | \off(e)\cap Z \neq \emptyset}\;\omega(e)$.
%
That is, $\PDsub{\Tree}(Z)$ is the total weight of all edges in $\Tree$ that are above some leaf in~$Z$.
%

The phylogenetic diversity model assumes that features of interest
appear along edges of the tree with frequency proportional to the weight of that edge, and that any feature belonging to one species is inherited by all its descendants. Thus,  $\PDsub{\Tree}(Z)$ corresponds to the expected number of distinct features appearing in all species in $Z$.

Initially defined only for trees, several extensions of the definition to phylogenetic networks have recently been proposed \cite{bordewich2022complexity,WickeFischer2018}.
In this paper, we focus $\NetPD$ (defined below), which allows the case that reticulations may not inherit all of the features from every parent.
This is modeled via an \emph{inheritance probability} $p(e) \in [0,1]$ on each reticulation edge $e = uv$. Here,~$p(e)$ represents the expected proportion of features present in $u$ that are also present in $v$; or equivalently, $p(e)$ is the probability that a feature in $u$ is inherited by $v$. Non-reticulation edges can be considered as having inheritance probability~$1$.

For a subset of taxa $Z \subseteq X$,
the measure $\NetPD(Z)$ represents the expected number of distinct features appearing in taxa in~$Z$~\cite{bordewich2022complexity}.
For each evolutionary branch~$uv$, this measure is obtained by multiplying
the number~$\w(uv)$ of features developed on the branch~$uv$ (which is assumed to be proportional to the length of the branch)
with the probability~$\gam{Z}(uv)$ that a random feature appearing in $u$ or developed on $uv$ will survive when preserving~$Z$.

Formally, we define~$\gam{Z}(uv)$ as follows. Consider an example in~\Cref{fig:gamma}.
\begin{figure}[t]
	\centering
	\begin{minipage}[t]{18ex}
		\begin{tikzpicture}[scale=.8,every node/.style={scale=.8}]
			\node[smallvertex, label=right:$\rho$] (root) at (0,0) {};
			\node[smallvertex, label=left:$v_1$] (v1) at ($(root) + (-1,-.75)$) {};
			\node[smallvertex, label=right:$v_2$, fill=black] (v2) at ($(v1) + (2,-.5)$) {};
			\node[smallvertex, label=right:$v_3$] (v3) at ($(v2) + (0,-.5)$) {};
			\node[smallvertex, label=left:$v_4$, fill=black] (v4) at ($(v1) + (0,-1.5)$) {};
			
			\draw[thick,-stealth] (root) to (v1);
			\draw[thick,-stealth] (root) to (v2);
			\draw[thick,-stealth] (v1) to (v2);
			\draw[thick,-stealth] (v1) to (v4);
			\draw[thick,-stealth] (v2) to (v3);
			\draw[thick,-stealth] (v3) to (v4);
			
			\nextnode[leaf, fill=white, label=right:$\ell_1$]{l1}{v3}{-90:.75}{revarc};
			\nextnode[leaf, fill=white, label=left:$\ell_2$]{l2}{v4}{-90:.75}{revarc};
		\end{tikzpicture}
	\end{minipage}\hfill
	\begin{minipage}[t]{63ex}
		\vspace{-3 cm}
		\begin{tabular}{lcccccccc}
			\toprule
			& $\rho v_1$ & $\rho v_2$ & $v_1 v_2$ & $v_2 v_3$ & $v_1 v_4$ & $v_3 v_4$ & $v_3 \ell_1$ & $v_4 \ell_2$ \\
			\midrule
			$\w(e)$         & 50    & 40    & 10    & 5     & 30    & 8     & 4     & 2     \\
			$p(e)$          & 1     & 0.4   & 0.5   & 1     & 0.2   & 0.6   & 1     & 1     \\
			\midrule
			$\gamma_{Z_1}^p(e)$ & 0.5   & 0.4   & 0.5   & 1     & 0     & 0     & 1     & 0     \\
			$\gamma_{Z_2}^p(e)$ & 0.44  & 0.24  & 0.3   & 0.6   & 0.2   & 0.6   & 0     & 1     \\
			$\gamma_{Z_3}^p(e)$ & 0.6   & 0.4   & 0.5   & 1     & 0.2   & 0.6   & 1     & 1     \\
			\bottomrule
		\end{tabular}
	\end{minipage}
	\caption{An example for calculating $\gamma_Z^p(e)$.
	Reticulations are black.
	The chosen sets are~$Z_1 = \{\ell_1\},Z_2 = \{\ell_2\},Z_3 = \{\ell_1,\ell_2\}$.
	$\NetPD^p(Z)$ for $Z=Z_1$, $Z_2$, $Z_3$ is 55, 50.4, and 72.8, respectively.
	}
	\label{fig:gamma}
\end{figure}%

\begin{definition}\label{def:gamma}    
Given a network $\Net = (V,E)$ with edge weights $\w:E\to\N$,
probabilities $p:E\to[0,1]$
and a set of taxa $Z\subseteq X$,
we define $\gam{Z}: E \to [0,1]$ recursively for each edge~$uv\in E$ as follows:
\begin{itemize}
  \item If $v$ is a leaf, then $\gam{Z}(uv) := p(uv)$ if $v \in Z$, and $\gam{Z}(uv) = 0$ otherwise.\\
    (\textbf{Intuition}: The features of $v$ survive if and only if $v$ is preserved by $Z$.)\\
    In most of the paper, with the notable exception of \Cref{sec:branching}, $p(uv)=1$ if $v$ is a leaf.
  \item If $v$ is a reticulation with outgoing arc $vx$, then $\gam{Z}(uv) = p(uv)\cdot\gam{Z}(vx)$.\\
    (\textbf{Intuition}: $v$'s features are a mixture of features of its parents and
    the features of $u$ have a certain probability $p(uv)$ of being included in this mix and, thereby, survive in preserved descendants of $x$.)
  \item If $v$ is a tree node with children $x_i$,
    then $\gam{Z}(uv) = 1 - \prod_i (1-\gam{Z}(vx_i))$.
    In the special case that $v$ has two children, this is equal 
    to $\gam{Z}(vx) + \gam{Z}(vy) - \gam{Z}(vx)\cdot\gam{Z}(vy)$.\\
    (\textbf{Intuition}: To lose a feature of $v$, it has to be lost in both children~$x$ and $y$ of $v$, which are assumed to be independent events, since both copies of the feature develop independently.)
\end{itemize}
\end{definition}
When clear from the context, we will omit the superscript~$p$.
Further, we only consider values of $p$ on edges incoming to leaves or reticulations, so we may restrict the domain of $p$
to those edges.
%
%
We can now define the measure $\NetPD^p(Z)$ for a subset of taxa $Z$ as follows:
    $\NetPD^p(Z) = \sum_{e \in E}\w(e)\cdot \gam{Z}(e)$.

Since we assume that all weights are non-negative, we observe that both $\gam{Z}(e)$ and $\NetPD(Z)$ are monotone on $Z$,
that is, $\gam{Z'}(e) \leq \gam{Z}(e)$ and $\NetPD^p(Z') \leq \NetPD^p(Z)$
for all $Z' \subseteq Z \subseteq X$.
We can now formally define the main problem studied in this paper:



\medskip
\begin{fbox}{
		\parbox{0.9\textwidth}{
			\MaxNPD\\
			\textbf{Input}:
      A phylogenetic network $\Net=(V,E)$ on $X$ with edge weights~$\omega:E\to\mathbb{N}$,
      inheritance probabilities~$p:E\to[0,1]$,
      and integers~$k,D\in\mathbb{N}$.\\
			\textbf{Question}:
      Is there a $Z\subseteq X$ with $|Z|\leq k$ and $\NetPD(Z) \geq D$?
	}}
\end{fbox}

\medskip\noindent
Note that, if $p(e)=1$ for all edges~$e$ incoming to leaves (all ``preservation projects'' succeed with probability~1)
and a node~$v$ has no reticulation descendants,
then $\gam[]{Z}(uv) = 1$ if $\off(e) \cap Z \neq \emptyset$, and otherwise $\gam[]{Z}(uv) = 0$ (see \Cref{lem:tree score}).
In this setting, $\NetPD$ coincides with $\PD$ if $\Net$ is a tree.
This holds even if all leaves are weighted and the total weight of $Z$ must not exceed~$k$.

Throughout the paper, we assume that integers are encoded in binary and that rational numbers $p/q$ (with $p$ and $q$ coprime integers) are encoded using binary encodings of $p$ and $q$. See~\Cref{sec:rationalEncodings} for details.


\section{A Branching Algorithm}\label{sec:branching}
\newcommand{\TPD}{\textsc{0/1-cost Max-Network-PD}\xspace}


In this section, we show that \MaxNPD is fixed-parameter tractable with respect to the reticulation number of the input network.
To facilitate the explanation of our algorithm, we solve a generalization of \textsc{Max-Network-PD}, where
(a)~each leaf~$\ell$ is assigned a cost~$c(\ell)\in\{0,1\}$,
(b)~the leaf-edges may have inheritance probability~$p(v\ell)\leq 1$ (as well as the reticulation edges), with the condition that $p(v\ell)=1$ if $c(\ell)=1$, and 
(c)~we look for a subset~$Z$ of leaves with total cost at most $k$ (instead of cardinality~$k$).
We refer to this problem as \TPD and
we use $p(\ell)$ instead of $p(v\ell)$ whenever $\ell$ is a leaf with parent~$v$.
%
%
%

In the following,
let $\Inst := (\Net, \w, p, c, k, D)$ be an instance for \textsc{0/1-cost Max-Network-PD},
let $r$ be a lowest reticulation in $\Net$ with outgoing edge~$rx$.
Our algorithm ``guesses'' whether or not any cost-1 leaf below~$r$ is in a solution $Z$.
If not, then we remove all cost-1 leaves below $r$ and use reduction rules to
(a)~turn the resulting subtree into a single leaf below $r$ and
(b)~turn $r$ into two new cost-0 leaves with inheritance probabilities according to $\gam[]{Z}(rx)$.
If some (unknown) cost-1 leaf below~$r$ is in a solution,
we show that such a leaf can be picked greedily.
Then, we decrement~$k$, set the cost of that leaf to zero, and use the knowledge that~$\gam[]{Z}(rx)=1$ to remove~$r$ from the network.

Note that our reduction and branching rules may create nodes with high outdegree, even if the input network is binary.
However, the algorithm used to solve the resulting non-binary tree can deal with such polytomies~\cite{PDinSeconds}.

\newcommand{\sprod}[1]{\ensuremath{\smashoperator{\prod_{#1}}\,}}

\paragraph*{Reduction.}
Let $r$ be a lowest reticulation in $\Net$ and let $E_r$ be the set of edges below $r$.
The following reduction rules simplify $\Inst$ by getting rid of cost-0 leaves below~$r$.
Note that each rule assumes that $\Inst$ is reduced with respect to the previous rules.
See \Cref{fig:reduction12} for examples of Reduction Rules~\ref{rr:deg2} and \ref{rr:zero p}, and
\Cref{fig:reduction34} for examples of Reduction Rules~\ref{rr:trivial reti} and \ref{rr:partial sol}.

\begin{figure}[t]
	\centering
	\begin{tikzpicture}[scale=.5,every node/.style={scale=0.8}]
		\foreach \j/\woff in {1/2, 2/1} {
			\node[smallvertex, label=right:$r$] (r\j) at (-6 + \j*6,0) {} edge ($(r\j)+(90-30:1)$) edge ($(r\j)+(90+30:1)$);
			\draw[fill=white] ($(r\j)+(0,1.5)$) -- ++(-1,.-1) -- ++ (2*1,0) -- cycle;
			\node[smallvertex, label=left:$\rho$] (root\j) at ($(r\j)+(0,1.5)$) {};
			
			\coordinate (xcoord\j) at ($(r\j) + (-90:1) + (-20:0)$);
			\node[smallvertex, label=right:$x$] (x\j) at (xcoord\j) {};
			\draw (x\j) -- ++(-60:1) -- ++(-180:1) -- (x\j); 
			
			\coordinate (ucoord\j) at ($(x\j) + (-60:1)$);
			\node[smallvertex, label=below:$u$] (u\j) at (ucoord\j) {};
			
			\coordinate (wcoord\j) at ($(u\j) + (-90:.5*\woff) + (0:\woff)$);
			\node[smallvertex, label=right:$w$] (w\j) at (wcoord\j) {};
			\draw (w\j) -- ++(-60:1) -- ++(-180:1) -- (w\j); 
		}
		\node at (3,0) {\scalebox{3}{$\leadsto$}};
		
		\node at (-1.25,1.5) {$(1)$};
		
		\coordinate (v1) at ($(u1) + (-90:.5) + (0:1)$);
		\node[smallvertex, label=below:$v$] (v1) at (v1) {};
		
		\draw[arc] (r1) -- (x1);
		\draw[arc] (r2) -- (x2);
		\draw[arc] (u2) to[bend left=40] (w2);
		\draw[arc] (u1) to[bend left=40] (v1);
		\draw[arc] (v1) to[bend left=40] (w1);

		\draw (9,2) -- (9,-3.7);
	\end{tikzpicture}
	\begin{tikzpicture}[scale=.5,every node/.style={scale=0.8}]
		\foreach \j/\woff in {1/2, 2/1} {
			\node[smallvertex, label=right:$r$] (r\j) at (-6 + \j*6,0) {} edge ($(r\j)+(90-30:1)$) edge ($(r\j)+(90+30:1)$);
			\draw[fill=white] ($(r\j)+(0,1.5)$) -- ++(-1,.-1) -- ++ (2*1,0) -- cycle;
			\node[smallvertex, label=left:$\rho$] (root\j) at ($(r\j)+(0,1.5)$) {};
			
			\coordinate (xcoord\j) at ($(r\j) + (-90:1) + (-20:0)$);
			\node[smallvertex, label=right:$x$] (x\j) at (xcoord\j) {};
			\draw (x\j) -- ++(-60:1) -- ++(-180:1) -- (x\j); 
			
			\coordinate (vcoord\j) at ($(x\j) + (-60:1)$);
			\node[smallvertex, label=right:$v$] (v\j) at (vcoord\j) {};
		}
		\node at (3,0) {\scalebox{3}{$\leadsto$}};
		
		\node at (-1.25,1.5) {$(2)$};
		
		\draw[arc] (r1) -- (x1);
		\draw[arc] (r2) -- (x2);
		
		\nextnode[leaf, fill=white, label=left:$\ell$]{l1}{v1}{-90:1.5}{revarc};
	\end{tikzpicture}
	\caption{Examples of Reduction Rules~\ref{rr:deg2} and \ref{rr:zero p} are depicted on the left and on the right, respectively.
	White leaves have an inheritance probability of zero.}
	\label{fig:reduction12}
\end{figure}
\todo{TS: I think that the sentence "White leaves have an inheritance probability of zero." is clear from the picture and can be removed. MW: but it doesn't hurt, does it?}
\begin{rrule}\label{rr:deg2}
  Let $uv\in E_r$ such that $v$ has a single child~$w$.
  Then, contract $v$ onto $u$ and set $\w(uw) := \w(uv)+\w(vw)$, $p(uw) := p(vw)$.
\end{rrule}
\begin{proof}[Correctness of \Cref{rr:deg2}]

  Let $\Inst' := (\Net', \w', p', c, k, D')$ be the result of applying \Cref{rr:deg2} to $\Inst$.
 Clearly, we have $\gam[p']{Z}(e) = \gam{Z}(e)$ for any edge $e$ below $w$ and all~$Z$.
 So by construction $\gam[p']{Z}(uw) = \gam{Z}(vw)$.
 Observe that $p(uv)=1$ since $v$ is not a leaf and~$r$ is the lowest reticulation in $\Net$; thus, $\gam{Z}(uv)=\gam{Z}(vw)=\gam[p']{Z}(uw)$.
 This implies that $\gam[p']{Z}(e) = \gam{Z}(e)$ for all~$Z$ and any $e \in E \setminus \{uv,vw\}$.
 So,  
 $\NetPD(Z) - \NetPDsub{\Net'}(Z) = \gam{Z}(uv)\cdot\w(uv) + \gam{Z}(vw)\cdot\w(vw) - \gam[p']{Z}(uw)\cdot w(uw) = \gam{Z}(uv)\cdot (w(uv) + w(vw) - w(uw))= 0$.
\end{proof}

\begin{rrule}\label{rr:zero p}
  Let $v\ell\in E_r$ such that $\ell$ is a leaf, $v\ne r$, and $p(v\ell)=0$.
  Then, remove~$\ell$.
\end{rrule}
\begin{proof}[Correctness of \Cref{rr:zero p}]
  Let $u$ be the parent of $v$, and let $v_1,\dots,v_t$ denote the children of $v$ with~$v_i\ne \ell$.
  Then $1 - (1-\gam{Z}(v\ell))\prod_i (1-\gam{Z}(uv_i))$ is the value of $\gam{Z}(uv)$ before removing~$\ell$ which equals~$1 - \prod_i (1-\gam{Z}(uv_i))$, the value afterward, since $\gam{Z}(v\ell)=0$ for all~$Z$.
\end{proof}

\begin{figure}[t]
	\centering
  \begin{tikzpicture}[scale=.5,every node/.style={scale=0.8}]
		\node at (-1.25,1.7) {$(3)$};
		
		\foreach \j in {2, 1} {
			\coordinate (r\j) at (-6 + \j*6,-1.5);
			\node[smallvertex, label=right:$r$] (r1) at (0,-1.5) {};
			\coordinate (root\j) at ($(r\j)+(0,3.199)$);
			\draw[fill=white] ($(root\j)$) -- ++(-1.9,-1.9) -- ++ (2*1.9,0) -- cycle;
			\node[smallvertex, label=left:$\rho$] (root\j) at (root\j) {};
			
			\node[smallvertex, label=above:$z_1$] (z1\j) at ($(r\j)+(90+30:1.5)$) {};
			\node[smallvertex, label=above:$z_2$] (z2\j) at ($(r\j)+(90:1.299)$) {};
			\node[smallvertex, label=above:$z_i$] (zi\j) at ($(r\j)+(90-30:1.5)$) {};
		}
		\node at (3,0) {\scalebox{3}{$\leadsto$}};
		
		\draw[arc] (z11) -- (r1);
		\draw[arc] (z21) -- (r1);
		\draw[arc] (zi1) -- (r1);
		
		\nextnode[leaf, fill=black, label=left:$x$, label=below:$0$]{x1}{r1}{-90:1}{revarc};
		\nextnode[leaf, fill=black, label=left:$\ell_{z_{1}}$, label=below:$0$]{x1}{z12}{-90:1}{revarc};
		\nextnode[leaf, fill=black,  label=below:$0$]{x1}{z22}{-90:1}{revarc};
		\nextnode[leaf, fill=black, label=right:$\ell_{z_{i}}$, label=below:$0$]{x1}{zi2}{-90:1}{revarc};
		
		\draw (9,2) -- (9,-3.5);
	\end{tikzpicture}
  \begin{tikzpicture}[scale=.5,every node/.style={scale=0.8}]
		\node at (-1.25,1.5) {$(4)$};
		
		\foreach \j/\woff in {1/2, 2/1} {
			\node[smallvertex, label=right:$r$] (r\j) at (-6 + \j*6,0) {} edge ($(r\j)+(90-30:1)$) edge ($(r\j)+(90+30:1)$);
			\draw[fill=white] ($(r\j)+(0,1.5)$) -- ++(-1,.-1) -- ++ (2*1,0) -- cycle;
			\node[smallvertex, label=left:$\rho$] (root\j) at ($(r\j)+(0,1.5)$) {};
			
			\coordinate (xcoord\j) at ($(r\j) + (-90:1) + (-20:0)$);
			\node[smallvertex, label=left:$x$] (x\j) at (xcoord\j) {};
			\draw (x\j) -- ++(1.4,-.866) -- ++(-180:2*1.4) -- (x\j); 
			
			\foreach \i in {0,1,2,3}
				\coordinate (xlow\j\i) at ($(x\j) + (-.9 + .6*\i,-.866)$);
		}
		\node at (3,0) {\scalebox{3}{$\leadsto$}};
		
		\draw[arc] (r1) -- (x1);
		\draw[arc] (r2) -- (x2);
		
		\foreach \i/\f/\l in {0/black/0, 1/black/1, 2/black/0, 3/black/0}
			\nextnode[leaf, fill=\f, label=below:$\l$]{l1\i}{xlow1\i}{-90:1}{revarc};
		\foreach \i/\f/\l in {0/white/0, 1/black/1, 2/white/0, 3/white/0}
			\nextnode[leaf, fill=\f, label=below:$\l$]{l2\i}{xlow2\i}{-90:1}{revarc};
		
		\node[leaf, fill=black, label=right:$\ell^*$, label=below:$0$] (ell) at ($(x2)+(1.9,-1.5)$) {};
		\draw[arc] (x2) to[bend left=40] (ell);
	\end{tikzpicture}
	\caption{Examples of Reduction Rules~\ref{rr:trivial reti} and \ref{rr:partial sol} are depicted on the left and on the right, respectively.
	Black leaves have a positive inheritance probability.
	Costs are written below the leaves.}
	\label{fig:reduction34}
\end{figure}
\todo{TS: Why is x in (3) black? It could also be a leaf with inheritance probability 0, especially after the application of Reduction Rule 1. MW: no, RR2 would have removed it if it had inheritence prob 0.}

\begin{rrule}\label{rr:trivial reti}
  Let the unique child~$x$ of $r$ be a leaf with cost $c(x)=0$.
  Then, for each parent~$z_i$ of $r$,
  add a new leaf~$\ell_{z_i}$ to $z_i$ with $c(\ell_{z_i}):=0$ and $p(\ell_{z_i}):=p(z_ir)\cdot p(rx)$ and $\w(z\ell_{z_i}):=\w(z_ir)$.
  Finally, remove $r$ and $x$ and decrease~$D$ by $p(rx)\cdot\w(rx)$.\todo{MW: technically, this requires $D$ to be a rational, but $D$ was defined to be int}
\end{rrule}
\begin{proof}[Correctness sketch of \Cref{rr:trivial reti}]
  As $x$ has cost 0 and $\gam{Z}$ is monotone 
  on
  $Z$, every maximal solution for $\Inst$ contains $x$.
  Likewise, every maximal solution for the modified instance~$\Inst'$ contains all~$\ell_{z_i}$.
  Then, one can verify that maximal solutions for $\Inst$ collect exactly the score of~$rx$ more than maximal solutions for $\Inst'$, which is $p(rx)\cdot\w(rx)$.
\end{proof}


\begin{rrule}\label{rr:partial sol}
  Let $x$ be the unique child of $r$,
  let $Q$ be the set of cost-0 leaves below~$r$ with $|Q|\geq 2$, and
  let $E_x:=E_r\setminus\{rx\}$.
  Then, 
  \begin{enumerate}[(1)]
    \item for each $uv\in E_x$, multiply $\w(uv)$ by $1-\gam{Q}(uv)$,
    \item for each $\ell\in Q$, set $p(\ell):=0$,
    \item reduce $D$ by $\sum_{e\in E_x}\gam{Q}(e)\cdot\w(e)$, and
    \item add a new cost-0 leaf $\ell^*$ as a child of $x$ with $\w(x\ell^*)=0$ and $p(\ell^*)=\gam{Q}(rx)$. 
  \end{enumerate}
\end{rrule}

To prove the correctness of \Cref{rr:partial sol}, we use the following lemma.

\newcommand{\smashprod}[1]{\ensuremath{\;\smashoperator{\prod_{#1}}\;}}

\begin{lemma}\label{lem:tree score}
  Let $uv$ be an edge in $\Net$ such that all descendants of $v$ (including~$v$) are tree nodes and
  let $Z$ be a leaf set of $\Net$.
  Then, $\gam{Z}(uv)=1-\prod_{\ell\in\off(uv)\cap Z}(1-p(\ell))$. 
\end{lemma}
\begin{proof}
    We prove the claim by induction on the length of a longest path from $v$ to a leaf.
    In the induction base, $v$ is a leaf and, thus,
    $\gam{Z}(e)=1-\prod_{\ell\in\off(e)\cap Z}(1-p(\ell))$
    since this is $1-(1-p(v)) = p(v)$ if $v\in Z$ and $0$ otherwise.
    For the induction step, let $v$ be a tree node with children $x_i$ and assume the claim is true for each edge $vx_i$.
    Then,
    \begin{align*}
        \gam{Z}(e) 
        & \stackrel{\text{Def.~\ref{def:gamma}}}{=} 1 - \prod_i \big(1-\gam{Z}(vx_i)\big)
        \stackrel{IH}{=} 1 - \prod_i \bigg(1 - \big(1- \smashprod{\ell\in\off(vx_i)\cap Z}(1-p(\ell))\big)\bigg)\\
        & = 1 - \prod_i \big(\;\;\;\;\smashprod{\ell\in\off(vx_i)\cap Z}(1-p(\ell))\big) =  1-\smashprod{\ell\in\off(e)\cap Z}(1-p(\ell))
        \tag*{\qed}\def\qed{}\qedhere
    \end{align*}
\end{proof}

\begin{proof}[Correctness of \Cref{rr:partial sol}]
  Let $\Inst' =: (\Net, \w', p', c, k, D')$ be the result of applying \Cref{rr:partial sol} to $\Inst$
  and let $Q':=Q\cup\{\ell^*\}$.
  We assume all solutions~$Z$ to be maximal, implying that they contain all cost-0 leaves.
  Note that generality is not lost since 
  $\NetPD^p(Z)$ is  monotone on $Z$.
  %
  Let $Z$ and $Z'$ be any subsets of leaves of $\Net$ and $\Net'$, respectively, with $Q\subseteq Z$ and $Z'=Z\cup\{\ell^*\}$.
  We show that $Z$ is a solution for $\Inst$
  if and only if $Z'$ is a solution for $\Inst'$.

  We consider the contribution of each edge to the diversity score of $Z$ in $\Net$ and the diversity score of $Z'$ in $\Net'$.
  If $Z$ (and, thus, also $Z'$) contains a cost-1 leaf~$\ell$ below~$r$, then $p(\ell)=1$ and, by \Cref{lem:tree score}, we have $\gam{Z}(rx)=1=\gam[p']{Z'}(rx)$.
  Otherwise, $\gam{Z}(rx)=\gam{Q}(rx)=p'(\ell^*)=\gam[p']{\{\ell^*\}}(rx)=\gam[p']{Z'}(rx)$.
  In both cases, $\gam{Z}(e)=\gam[p']{Z'}(e)$ for all $e\in E\setminus E_r$ since these values only depend on the values of the edges below~$e$.
  %
  %
  Further, note that $\w(x\ell^*)\cdot \gam[p']{Z'}(x\ell^*)=\w(x\ell^*)\cdot p(\ell^*)=0$.
  Thus, it remains to consider the edges in $E_x:=E_r\setminus\{rx\}$.
  For any such edge $e\in E_x$, we observe
  \begin{align*}
    \gam{Z}(e)
    & 
    \stackrel{\text{\Cref{lem:tree score}}}{=} 1 - \sprod{\ell\in\off(e)\cap Z}(1-p(\ell))
    = 1 - \sprod{\ell\in\off(e)\cap Q}(1-p(\ell))\cdot\sprod{\ell\in\off(e)\cap Z \setminus Q}(1-p(\ell))\\
  \intertext{and, since $p(\ell)=1$ for all $\ell\in Z\setminus Q$ by convention stated in the problem definition, we have}
    \gam{Z}(e)
    &
    =\begin{cases}
      \underbrace{1-\sprod{\ell\in\off(e)\cap Q} (1-p(\ell))}_{\gam{Q}(e)} & \text{if }\off(e)\cap Z\subseteq Q\\
      1 & \text{otherwise}
    \end{cases}
    \intertext{and the same holds for $p'$ and $Z'$ instead of $p$ and $Z$ since the leaves below $e$ in $\Net$ are exactly the leaves below $e$ in $\Net'$ ($\ell^*$ cannot be below $e$ in $\Net'$ since $e\in E_x$). Now, since $p'(\ell)=0$ for all $\ell\in Q$ by construction, we have}
    \gam[p']{Z'}(e)
    &
    = \begin{cases}
      0 & \text{if } \off(e)\cap Z\subseteq Q\\
      1 & \text{otherwise}
    \end{cases}
  \end{align*}
  implying $\gam[p]{Z}(e) = \gam[p']{Z'}(e) \cdot (1 - \gam{Q}(e)) + \gam{Q}(e)$.
  Then,
  \begin{align*}
    \sum_{e\in E} \gam{Z}(e) \cdot \w(e) - \smashoperator{\sum_{e\in E\cup\{x\ell^*\}}}\gam[p']{Z'}(e)\cdot\w'(e)
    & = \sum_{e\in E_x}(\gam{Z}(e)\cdot\w(e) - \gam[p']{Z'}(e)\cdot\w'(e))\\
    & \stackrel{\text{Def'n~$\w'$}}{=} \sum_{e\in E_x} \big(\gam{Z}(e) - \gam[p']{Z'}(e)\cdot (1-\gam{Q}(e)) \big)\cdot \w(e)\\
    & = \sum_{e\in E_x} \gam{Q}(e) \cdot \w(e) =  D - D'
  \end{align*}
  Thus, $\sum_{e\in E}\gam{Z}(e)\cdot\w(e)\geq D$ if and only if $\sum_{e\in E}\gam[p']{Z}(e)\cdot\w(e)\geq D'$.
\end{proof}

\paragraph*{Branching.}
Observe that, if no reduction rule applies to $\Net$,
then the subtree below any lowest reticulation~$r$ has at least one cost-1 leaf and at most one cost-0 leaf.
An important part of the correctness of our branching algorithm is that solutions may be assumed to pick cost-1 leaves ``greedily'',
that is, if a solution chooses any cost-1 leaf below~$r$, then there is also a solution choosing a ``heaviest'' cost-1 leaf below~$r$ instead.

\begin{lemma}\label{lem:greedy+branching}
  Let~$r$ be a lowest reticulation in $\Net$ and
  let~$a$ be some cost-1 leaf below~$r$ in $\Net$ maximizing the weight of the $r$-$a$-path.
  Let $Z$ be any set of leaves of $\Net$ containing a cost-1 leaf below~$r$.
  Then, there is a set~$Z^*$ of leaves of $\Net$ with the same cost as $Z$ with $a\in Z^*$ and $\NetPD^p(Z^*)\geq \NetPD^p(Z)$.
\end{lemma}
\begin{proof}
  Suppose that $a\notin Z$ as otherwise, the claim is trivial.
  Let $b\in Z$ be a cost-1 leaf below~$r$ such that $u:=LCA(a,b)$ is lowest possible (has maximal (unweighted) distance from~$r$), and
  let $Z^*:=(Z\setminus\{b\})\cup\{a\}$. 
  Let $q_a$ and $q_b$ be the unique paths from $u$ to $a$ and $b$, respectively, and 
  note that $\w(q_a)\geq \w(q_b)$ by choice of $a$.
  Furthermore, for each edge~$uv$ on $q_a$, we know that $Z$ contains no leaf below~$v$ (by maximality of the $r$-$u$-path).
  Since both $a$ and $b$ are cost-1 leaves, we have $p(a)=p(b)=1$ by convention stated in the problem definition,
  implying that $\gam{Z}(e_b)=\gam{Z^*}(e_a)=1$ for all edges $e_a$ on $q_a$ and $e_b$ on~$q_b$.
  Thus, $\NetPD^p(Z^*) - \NetPD^p(Z) = \sum_{e\in E}\gam{Z^*}(e)\w(e) - \sum_{e\in E}\gam{Z}(e)\w(e) = \w(q_a) - \w(q_b) \geq 0$.
\end{proof}

\begin{figure}[t]
  \centering
  \begin{tikzpicture}[scale=.5,every node/.style={scale=0.8}]
    \foreach \j/\xoff in {0/1, 1/0, 2/1} {
      \node[smallvertex, label=right:$r$] (r\j) at (-8 + \j*8,0) {} edge ($(r\j)+(90-30:1)$) edge ($(r\j)+(90+30:1)$);
      \draw[fill=white] ($(r\j)+(0,2)$) -- ++(-1.2,.-1.5) -- ++ (2*1.2,0) -- cycle;
      \node[smallvertex, label=left:$\rho$] (root\j) at ($(r\j)+(0,2)$) {};

      \coordinate (xcoord\j) at ($(r\j) + (-90:1) + (0:\xoff*2)$);
      \coordinate (xlow\j) at ($(xcoord\j) + (-135:2)$);
      \node[smallvertex, label=right:$x$] (x\j) at (xcoord\j) {};
      \draw (x\j) -- ++(-45:2) -- (xlow\j) -- (x\j); 
    }
    \node at (-4,0) {\scalebox{3}{\reflectbox{$\leadsto$}}};
    \node at (4,0) {\scalebox{3}{$\leadsto$}};

    \foreach \i/\f/\c in {1/white/0, 2/white/0, 3/white/0, 4/white/0}
      \node[smallleaf, fill=\f, label=below:$\c$] (l\i) at ($(xlow0)+(\i*.6,-1)$) {} edge[revarc] ($(xlow0)+(\i*.6,0)$);
    \foreach \i/\f/\c in {1/black/1, 2/black/1, 3/black/1, 4/white/0}
      \node[smallleaf, fill=\f, label=below:$\c$] (l\i) at ($(xlow1)+(\i*.6,-1)$) {} edge[revarc] ($(xlow1)+(\i*.6,0)$);
    \foreach \i/\f/\c in {1/black/1, 2/black/0, 3/black/1, 4/white/0}
      \node[smallleaf, fill=\f, label=below:$\c$] (l\i) at ($(xlow2)+(\i*.6,-1)$) {} edge[revarc] ($(xlow2)+(\i*.6,0)$);

    \draw[arc] (r1) -- (x1);

    \foreach \j/\f in {0/white,2/black}{
      \draw[arc] (root\j) to[bend left=40] (x\j);
      \nextnode[leaf, fill=\f, label=left:$\ell$, label=below:$0$]{l\j}{r\j}{-90:1}{revarc}
    }

  \end{tikzpicture}
  \caption{An example of \Cref{br:main} with
  $\Inst_0$ (``do not select a cost-1 leaf below~$r$'') on the left and
  $\Inst_1$ (``select a cost-1 leaf below~$r$'') on the right.
  Black leaves have an inheritance probability of one.
  Costs are written below the leaves.
  Note that the budget for $\Inst_1$ is $k-1$ and that applying \Cref{rr:partial sol} may change the target diversity.}
  \label{fig:branching}
\end{figure}

Now, we can present and prove the correctness of our main branching rule, solving \TPD in $O^*(\binom{|R|}{k})$~time,
where $R$ is the set of reticulations in the input network and $k$ is the budget. 

\begin{brule}[See \Cref{fig:branching}]\label{br:main}
  \looseness=-1
  Let $\rho$ be the root of $\Net$.
  Let~$r$ be a lowest reticulation in~$\Net$ whose unique child~$x$ is not a 0-cost leaf.
  Let $Q$ be the set of cost-0 leaves below~$r$.
  Then,
  \begin{enumerate}[1.]
    \item create the instance~$\Inst_0:=(\Net_0,\w_0,p_0,c_0,k,D)$ by
      \begin{enumerate}[(a)]
        \item setting $p_0(t):=0$ and $c(t):=0$ for all cost-1 leaves~$t$ below~$r$,
        \item replacing $rx$ with $\rho{}x$, setting $\w_0(\rho{}x):=\w(rx)$ and,
        \item adding a new leaf~$\ell$ to $r$ with $p_0(\ell):=\gam{Q}(rx)$ and $c_0(\ell):=\w_0(r\ell):=0$, and
      \end{enumerate}
    \item create the instance~$\Inst_1:=(\Net_1,\w_1,p_1,c_1,k-1,D)$ by
      \begin{enumerate}[(a)]
        \item finding a cost-1 leaf~$a$ below $r$ maximizing the weight of the $r$-$a$-path and setting $c_1(a):=0$,
        \item replacing $rx$ with $\rho{}x$, setting $\w_1(\rho{}x):=\w(rx)$ and
        \item adding a new leaf $\ell$ to $r$ with $p_1(\ell):=1$ and $c_1(\ell):=\w_1(r\ell):=0$.
      \end{enumerate}
  \end{enumerate}
\end{brule}
\begin{proof}[Correctness of \Cref{br:main}]
  Let $P$ denote the set of cost-1 leaves below~$r$ in $\Inst$ and
  recall that $Q$ contains all cost-0 leaves below~$r$ in $\Inst$, and that $c(Q)=c_0(Q)=0$.
  We show that $\Inst$ has a solution~$Z$ if and only if $\Inst_0$ or $\Inst_1$ has a solution.
  Without loss of generality, we may assume solutions to be maximal, that is, they contain all cost-0 leaves.
  For any leaf-set~$Z$ containing all cost-0 leaves in $\Inst$ and
  any leaf-set~$Z_i$ containing all cost-0 leaves in~$\Inst_i$ for some~$i\in\{0,1\}$,
  we then have
  \begin{align}\hspace{-3.5ex} 
    \gam[p_i]{Z_i}(\rho{}x)
    = \begin{cases}
        \gam{Q}(rx) & \text{if $i=0$}\\
        1 & \text{if $i=1$}
      \end{cases} = p_i(\ell) = \gam[p_i]{Z_i}(r\ell)
    \;\text{ and }\;
    \gam{Z}(rx) = \begin{cases}
      \gam{Q}(rx) & \text{if $Z\cap P=\emptyset$}\\
      1 & \text{if $Z\cap P\ne \emptyset$}
    \end{cases}
    \label{eq:rho and x}
  \end{align}
  so, under the condition $Z \cap P = \emptyset \iff i = 0$, we have $\gam[p_i]{Z_i}(\rho{}x)=\gam{Z}(rx)$, so
  \begin{align}
    \gam[p_i]{Z_i}(\rho{}x)\cdot \underbrace{\w_i(\rho{}x)}_{=\w(rx)} + \gam[p_i]{Z_i}(r\ell)\cdot\underbrace{\w_i(r\ell)}_{=0}
    \stackrel{\eqref{eq:rho and x}}{=} \gam{Z}(rx)\cdot\w(rx).
    \label{eq:diff gamma}
  \end{align}

  \begin{claim}\label{cl:ZZ_i correct}
    Let $Z$ be a leaf-set in $\Net$,
    let $Z':=Z\cup\{\ell\}$, and
    let $i:=\operatorname{sgn}(|Z \cap P|)$.
    Then, $Z$ is a solution for $\Inst$ if and only if $Z'$ is a solution for $\Inst_i$.
  \end{claim}
  \begin{claimproof}
    Note that $Z \cap P = \emptyset \iff i=0$ is satisfied.
    In the following, we compare the value of $Z$ in $\Inst$ and
    the value of $Z'$ in $\Inst_i$.

    First, consider any arc~$e$ in $\Net$ that is not below~$r$.
    Since, by \eqref{eq:rho and x}, we have $\gam{Z}(rx)=\gam[p_i]{Z'}(r\ell)$,
    and since $p(\ell') = p_0(\ell')$ for any leaf $\ell' \neq \ell$,
    we inductively infer that $\gam[p_i]{Z'}(e)=\gam{Z}(e)$
    as these values only depend on the edges below~$e$.

    Second, by \eqref{eq:diff gamma},
    the contribution of the arc~$rx$ to the value of the solution~$Z$ for $\Inst$ equals
    the contribution of~$\rho{}x$ and~$r\ell$ to the value of the solution~$Z'$ for $\Inst_i$.
    
    It remains to compare the contributions of the arcs~$e$ below $x$ in $\Net$.
    In the following, consider such an arc~$e$.
    If $i=0$, then $Z$~avoids~$P$ and so does~$Z'$, so~$p(\ell')=p_0(\ell')$ for all $\ell'\in\off(e)\cap Z = \off(e)\cap Z'$.
    If $i=1$, then $p(\ell')=p_1(\ell')$ for all leaves in~$\off(e)$.
    Thus, by \Cref{lem:tree score},
    \begin{align*}
      \gam[p_i]{Z'}(e)
      = 1-\sprod{\ell'\in\off(e)\cap Z'}(1-p_i(\ell')) 
      = 1-\sprod{\ell'\in\off(e)\cap Z}(1-p(\ell'))
      = \gam{Z}(e).
    \end{align*}
    Thus, we conclude that $Z$ and $Z'$ score exactly the same in $\Inst$ and $\Inst'$, respectively.
    
    Finally, we show that $c(Z)=c_i(Z')-i$.
    If~$i=0$, then this holds since~$c_0(\ell)=0$.
    If~$i=1$ then $Z$ intersects~$P$ and, by \Cref{lem:greedy+branching}, we can assume that $Z$ contains~$a$.
    Then, since~$c(a)=1$ and~$c_1(\ell)=c_1(a)=0$, we have $c_1(Z')=c(Z)-1$.
  \end{claimproof}

  Now, we can prove the promised equivalence.
  First, if $Z$ is a solution for $\Inst$, then $Z':=Z\cup\{\ell\}$ is a solution
  for $\Inst_i$ with $i=\operatorname{sgn}(|Z\cap P|)$.
  Second, if $Z_0$ is a solution for $\Inst_0$, then $Z_0':=Z_0\setminus P$ is also a solution for $\Inst_0$
  since $p_0(\ell')=0$ for all $\ell'\in P$ and, by \Cref{cl:ZZ_i correct}, $Z:=Z_0'\setminus\{\ell\}$ is a solution for $\Inst$.
  Third, if $Z_1$ is a solution for $\Inst_1$ then we can assume $a\in Z_1$ since $c_1(a)=0$ so, for $Z:=Z_1\setminus\{\ell\}$,
  we have $Z\cap P\ne\emptyset$, thereby satisfying the conditions of \Cref{cl:ZZ_i correct}.
  Thus, $Z$ is a solution for $\Inst$.
\end{proof}

We can now solve \TPD as follows.
If $k = 0$, then the monotonicity of $\NetPD^p(Z)$ in $Z$ implies that ``taking'' all cost-0 leaves in~$\Net$ is optimal.
Otherwise, we repeatedly find a lowest reticulation~$r$ in $\Net$,
apply all reduction rules, and if $r$ survives, branch into two instances using \Cref{br:main}.
Note that, in each new instance, $r$~has a leaf child with cost~$0$.
Thus, \Cref{rr:trivial reti} will apply and remove~$r$ before another branching occurs. 
If no branching or reduction rules apply, then $\Net$ is a tree.
In this tree, a slight variation of \Cref{rr:partial sol} can be used to remove all cost-0 leaves,
so all remaining leaves have cost~1 and, therefore (by convention), inheritance probability~1.
Such an instance can be solved in $\Oh(n \log k)$~time~\cite{PDinSeconds}.
%
Note that the budget~$k$ is decreased for one of the two branches and $|R|$ is reduced in each branch, so no more than~$\binom{|R|}{k}$ branches need to be explored.
Finally, with 
careful bookkeeping
the reduction and branching can be implemented to run in~$\Oh(|E|) = \Oh(n+r)$ amortized time in total.

\begin{theorem}\label{thm:branching}
  On binary, $n$-leaf networks with~$r$ reticulations,
  \TPD and \MaxNPD can be solved in 
  $\Oh(\sum_{i=0}^{\min\{k,r\}}\binom{r}{i}\cdot \log k\cdot (n+r))\subseteq \Oh(2^r\cdot \log k\cdot (n+r))$~time,
  where $k$ is the budget.\footnote{Note that this running time degenerates to $o(2^r\cdot n)$ if $k\leq r/3$}
\end{theorem}

\Cref{thm:branching} shows that \MaxNPD is fixed-parameter tractable with respect to the number of reticulations.
In light of this, one might expect that \MaxNPD is also fixed-parameter tractable with respect to the ``level''
(maximum number of reticulations in any biconnected component (``blob'') of the network,
since many tractability results for the reticulation number also extend to the level
by applying the algorithm separately to each blob, with minimal adjustment,
in such a way that the problem parameterized by level reduces to the problem parameterized by reticulation number.
Unfortunately, this approach does not work for \MaxNPD -- for a given blob,
it may be better to pay some diversity score within the blob in order to increase~$\gam{Z}(e)$ for the incoming edge of that blob.
This trade-off means that we need to consider many possible solutions for each blob.
Indeed, we will see in the next section that \MaxNPD is NP-hard even on level-1 networks.



\section{NP-hardness Results}\label{sec:reduction}

Complementing the positive result of the previous section, we now show that \MaxNPD is \NP-hard on level-1 networks, answering an open question in the literature~\cite[Section~9]{bordewich2022complexity}.
On our way to showing this hardness result,
we also show \NP-hardness of the following problem,
answering an open question of \citet{komusiewicz2023multivariate}:

\medskip
  \begin{fbox}{
\parbox{0.9\textwidth}{
  \ucNAP\\
  \textbf{Input}: A tree $\Tree=(V,E)$ with leaves~$L$, edge weights~$\omega:E\to\mathbb{N}$, success probabilities~$p:L\to[0,1]$, and some~$k,D\in\mathbb{N}$.\\
  \textbf{Question}: Is there some $Z\subseteq L$ with $|Z|\leq k$ and $\sum_{e\in E}\gamma'_Z(e)\cdot\omega(e)\geq D$, where $\gamma'_Z(e) := (1 - \prod_{x \in \off(e) \cap Z}(1-p(x)))$?
  }}
  \end{fbox}

\medskip\noindent
Note that $\gamma'_Z(e)$ corresponds to the probability that at least one taxa in $\off(e)$ survives, under the assumption that every taxon~$x \in Z$ survives independently with probability $p(x)$, and every taxon~$x \in L \setminus Z$ does not survive.
Thus, \ucNAP can be viewed as the problem of maximizing the expected phylogenetic diversity on a tree, where each species we choose to save has a certain probability of surviving.


\mysubparagraph{Subset Product.} 
First, we show that the following problem is \NP-hard.

\begin{fbox}{
\parbox{0.9\textwidth}{
  \textsc{\SubProd}\\
  \textbf{Input}: A multiset of positive integers $\{v_1,v_2,\dots, v_m\}$, integers $M,k\in\mathbb{N}$.\\
  \textbf{Question}: Is there any $S \subseteq [m]$ with $|S|=k$ such that $\prod_{i \in S}v_i = M$?
  }}
  \end{fbox}

\medskip\noindent
We note that the definition of \SubProd is slightly different here from the formulation of \citet{GareyJohnson79}. In particular, we assume that the size~$k$ of the set~$S$ is given and that all integers are positive. This makes the subsequent \NP-hardness reductions in this paper slightly simpler.

The \NP-hardness of \SubProd is not a new result. It was stated by~\citet{GareyJohnson79} without full proof (the authors indicate that the problem is \NP-hard by reduction from \XTClong (\XTC), citing ``Yao, private communication'') and a full proof appears in~\cite{Moret97} and we reproved it for our slightly adapted variant in the appendix.

\begin{lemma}[\cite{Moret97}]\label{lem:XTCtoSubProb}
  \XTC reduces to \SubProd in polynomial time.
\end{lemma}

\noindent
As \XTC is \NP-hard~\cite{GareyJohnson79}, so is \SubProd.


\mysubparagraph{Penalty Sum.}
\citet[Theorems~5.3~\&~5.4]{komusiewicz2023multivariate} showed that,
if the following problem is \NP-hard, then so is \ucNAP:

\begin{fbox}{
  \parbox{0.9\textwidth}{
  \PenSum\\
  \textbf{Input}: A set of tuples $\{t_i = (a_i,b_i) \mid i \in [m], a_i \in \mathbb{Q}_+\cup \{0\}$, $b_i \in \mathbb(0,1)\}$, integers~$k$,~$Q$, and a number $D \in \mathbb{Q}_+$.\\
  \textbf{Question}: Is there some $S\subseteq [m]$ with $|S| = k$ such that  $\sum_{i\in S}a_i - Q\cdot \prod_{i\in S}b_i\geq D$?
  }}
\end{fbox}

\medskip\noindent
\iflnbi
\else
We set out to show \PenSum~\NP-hard by reducing \SubProd to it.
To communicate the main ideas of this reduction, we first describe a simple transformation that turns an instance of \SubProd into an equivalent `instance' of \PenSum, but one in which the numbers involved are irrational (and as such, cannot be produced in polynomial time). 
We then show how this transformation can be turned into a polynomial-time reduction by replacing the irrational numbers with suitably chosen rationals.


Finally, we reduce \ucNAP to \MaxNPD on level-1 networks,
showing that this restriction of \MaxNPD is also \NP-hard.
\fi

%
%
%
%

\subsection{Hardness of Penalty Sum}\label{sec:reduction-PenSum}
\iflnbi
\else
  \input{psum_hard.tex}

\fi

\subsection{Hardness of Network-Diversity}\label{sec:reduction-MaxNPD}
\input{nd_hard.tex}







\section{Discussion}

In this paper, we have studied \MaxNPD{} from a theoretical point of view. These results do have some practical implications. In particular, they show that we can only hope to solve \MaxNPD{} efficiently for evolutionary histories that are reasonably tree-like in the sense that the number of reticulate events is small. For this case, we present an algorithm that is theoretically efficient. How well it works in practice is still to be evaluated. 

Some open questions on the theoretical front remain.
Can \MaxNPD{} be solved in pseudo-polynomial time on level-1 networks?
Is \MaxNPD{} polynomial time solvable on level-1 networks if we require the network to be ultrametric, i.e. when all root-leaf paths have the same length?
Is \MaxNPD{} FPT when parameterized with the number of selected species~$k$ to save plus the level of the network?
Is \MaxNPD FPT when parameterized with the number of different weights and probabilities?
If this is the case, then rounding the weights or probabilities to magnitudes could significantly speed up the running time~\citep{numnum}.

From a practical point-of-view however, the most important task is to assess which variants of phylogenetic diversity on networks (see~\cite{WickeFischer2018}) are biologically most relevant. This could of course depend on the type of species considered and in particular on the type of reticulate evolutionary events. Even if the maximization problem cannot be solved efficiently, having a good measure of phylogenetic diversity can still have great practical use by measuring how diverse a given set of species is.


\appendix

\section{A note about binary representation of rational numbers}\label{sec:rationalEncodings} 

As most of the problems here involve rational numbers as part of the input, it is worth drawing attention to how those numbers are represented, in particular how they affect the input size of an instance.
As is standard, we assume that postitive integers are represented in binary (so that, for instance, the numbers 3, 4 and 5 are written as \texttt{11}, \texttt{100} and \texttt{101} respectively).
 Thus the number of bits required to represent the integer $n$ is $\Oh(\log_2(n))$.
 In the case of rational numbers, we assume throughout that a rational $p/q$ (with $p$ and $q$ coprime integers) can be represented by binary representations of $p$ and $q$. Thus for example, the number $3/5$ may be written as \texttt{11/101}. It follows that $p/q$ can be represented using $\Oh(\log_2(p) + \log_2(q))$ bits.

 For rational numbers which are a multiple of a power of 2, (such as $1/8 = 2^{-3}$, or $5/8 = 5\cdot2^{-3}$), we can write the number by extending the binary representation 'past the decimal point', so that e.g. $1/8$ would be written as \texttt{0.001} and $5/8$ as \texttt{0.101}.
 There is also the 'floating point' representation, where the number is expressed as an integer $t$ times $2$ to some integer $c$, and the numbers $t$ and $c$ are expressed in binary. Thus for example $5/8$ would be written as $\texttt{101} \times 2^\texttt{-11}$.
Both of these methods of representing rationals have the drawback that they cannot represent rationals that are not a multiple of a power of $2$. The number $1/3$, for instance, cannot be expressed exactly under either method.

 This distinction becomes important in \Cref{sec:reduction-PenSum}, where our reduction from \SubProd{} to \PenSum{} produces rational numbers that are not multiples of a power of $2$. 
Do our hardness results for \PenSum, \ucNAP and \MaxNPD{} still hold when one insists on a different method of representing rationals? This is an interesting question, and we make no attempt to answer it.

\section{Omitted Proofs}

To prove \Cref{lem:XTCtoSubProb}, we reduce the following problem to \SubProd.

\begin{fbox}{
\parbox{0.9\textwidth}{
 \XTClong (\XTC)\\
  \textbf{Input}: A set $X$ with $|X|=3n$, a collection $\mathcal{C}$ of subsets of $X$ with $|C|=3$ for every $C \in \mathcal{C}$.\\
  \textbf{Question}: Is there a collection $\mathcal{C}'\subseteq \mathcal{C}$ such that each element of $x$ appears in exactly one set of $\mathcal{C}'$?
  }}
\end{fbox}

\begin{proof}[Proof of \Cref{lem:XTCtoSubProb}]
	Let $(X := \{x_1,\dots, x_{3n}\},\mathcal{C} := \{C_1,\dots, C_m\})$ be an instance of \XTC.
	Let $p_1,\dots, p_{3n}$ be the first $3n$ prime numbers, so that we may associate each $x_j \in X$ with a unique prime number $p_j$.
	For each set $C_i = \{x_a, x_b, x_c\}$, let $v_i := p_a\cdot p_b \cdot p_c$, that is, $v_i$ is the product of the three primes associated with the elements of $C_i$.
	Now let $M := \prod_{j=1}^{3n} p_j$, i.e. $M$ is the product of the prime numbers $p_1,\dots, p_{3n}$.
	Finally let $k = n$.
	This completes the construction of an instance  $(\{v_1,\dots v_m\}, M, k)$ of \SubProd. 

	Now observe that if $\prod_{i \in S}v_i = M$ for some $S \subseteq [m]$, then by uniqueness of prime factorization, every prime number $p_1,\dots, p_m$ must appear exactly once across the prime factorizations of all numbers in $\{v_i:i\in S\}$. It follows by construction that the collection of subsets $\mathcal{C}' := \{C_i: i \in S\}$ contains each element of $X$ exactly once. Thus, if $(\{v_1,\dots v_m\}, M, k)$ is a \yes-instance of \SubProd then 
	$(X,\mathcal{C})$ is a \yes-instance of \XTC.
	Conversely, if $(X,\mathcal{C})$ is a \yes-instance of \XTC with solution $\mathcal{C'}$, then we can define $S:= \{i \in [m]: C_i \in \mathcal{C}'\}$. Since every element of $X$ appears in exactly one $C_i \in \mathcal{C}'$ and $|C_i| = 3$ for all $i \in [m]$, we have that $|\mathcal{C}'| = |X|/3 = n = k$, and $\prod_{i \in S}v_i = p_1\cdot\dots\cdot p_{3n} = M$. Thus $(\{v_1,\dots v_m\}, M, k)$ is a \yes-instance of \SubProd.

	It remains to show that the construction of $(\{v_1,\dots v_m\}, M, k)$ from  $(X,\mathcal{C})$ takes polynomial time.
	In particular, we need to show that each of the primes $p_1,\dots p_{3n}$ (and thus the product $M$) can be constructed in polynomial time. This can be shown using two results from number theory: 
	$p_j <  j(\ln j + \ln \ln j)$ for $j\geq 6$, \cite{Rosser41, Dusart99} 
	and the set of all prime numbers in $[Z]$ can be computed in time $\Oh(Z/ \ln \ln Z)$~\cite{AtkinBernstein04}. 
	Combining these, we have that the first $3n$ prime numbers can be generated in time  
	$\Oh(n \ln n / \ln \ln n)$.

	Given the prime numbers $p_1,\dots, p_{3n}$, it is clear that the numbers $\{v_i: i\in [m]\}$ can also be computed in polynomial time. The number $M$, being the product of $3n$ numbers each less than $3n(\ln 3n + \ln \ln 3n)$, can also be computed in time polynomial in $n$ (though $M$ itself is not polynomial in $n$).
	It follows that $(\{v_1,\dots v_m\}, M, k)$ can be constructed in polynomial time.
\end{proof}

\begin{proof}[Proof of \Cref{lem:logDifferenceBound}]
	We first show that it is enough to consider the cases $Q' = Q+1$ and $Q' = Q-1$.
	Fix an integer $Q \in \mathbb{N}_+$ with $Q \geq 2$.
	Consider the function $h_Q:\mathbb{R}_{>0}\rightarrow \mathbb{R}$ given by
	$$h_Q(x) = \ln x - \ln Q + Q/x - 1.$$

	So our aim is to show that $h_Q(Q') \geq Q^{-4}$.
	Similar to the proof of \Cref{lem:PenSumIrrationalMax}, we can observe that 
	$$\frac{dh_Q}{dx} = x^{-1} - Qx^{-2} = \frac{1}{x}\left(1- \frac{Q}{x}\right)$$
	is less than $0$ when $x < Q$, exactly $0$ when $x = Q$, and greater than $0$ when $x>Q$.
	It follows that on the range $x>0$, the function $h_Q$ has a unique minimum at $x = Q$, and is decreasing on the range $x< Q$ and increasing on the range $x>Q$.
	Thus in particular $h_Q(Q') \geq h_Q(Q-1)$ if $Q' \leq Q-1$ and $h_Q(Q') \geq h_Q(Q+1)$ if $Q' \geq Q+1$.
	Since either $Q' \leq Q-1$ or $Q' \geq Q+1$ for any integer $Q'\neq Q$, it remains to show that 
	$h_Q(Q-1) > Q^{-4}$ and  $h_Q(Q+1) > Q^{-4}$.
	
	To show $h_Q(Q-1) > Q^{-4}$ for any $Q \in \mathbb{N}_{\ge 2}$: Let $\lambda:\mathbb{R}_{>0}\rightarrow\mathbb{R}$ be the function given by 
	\begin{eqnarray*}
	\lambda(Q) & = & h_Q(Q-1) - Q^{-4}\\
	&=& \ln(Q-1) - \ln Q + Q/(Q-1) - 1 - Q^{-4}\\
	&=& \ln(Q-1) - \ln Q + 1/(Q-1) - Q^{-4}.
	\end{eqnarray*}

	Then 
	\begin{eqnarray*}
		\frac{d\lambda}{dQ} &=& (Q-1)^{-1} - Q^{-1} + (Q-1)^{-2} + 4Q^{-5}\\
		&>& (Q-1)^{-2} + 4Q^{-5} \\
		&>& 0.
	\end{eqnarray*}

	It follows that $\lambda$ is a (strictly) increasing function.
	Since $\lambda(2) = 0 - \ln 2 + 1 - 1/16 \approx 0.244 > 0$, it follows that $\lambda(Q) > 0$ for all $Q\geq 2$, and thus $h_Q(Q-1) > Q^{-4}$.

	To show that $h_Q(Q+1) > Q^{-4}$ for all $Q \in \mathbb{N}_{\geq 2}$:
	First observe that if $Q = 2$, then $h_Q(Q+1) = \ln(3) - \ln(2) + 2/3 - 1 \approx 0.0721 > 0.0625 = 2^{-4}$ and so the claim is true.
	For $Q \geq 3$, observe that $h_Q(Q+1) = \ln(Q+1) - \ln Q + \frac{Q}{Q+1} - 1 = \ln(\frac{Q+1}{Q}) - \frac{1}{Q+1}$.
	We use the Mercator series for the natural logarithm:
	\begin{eqnarray*}
		\ln\left(\frac{Q+1}{Q}\right)
		= \ln\left(1+\frac{1}{Q}\right)  
		= \sum_{k = 1}^\infty \frac{(-1)^{k+1}}{kQ^k}
		= \frac{1}{Q} - \frac{1}{2Q^2} + \frac{1}{3Q^3} - \frac{1}{4Q^4} + \dots
	\end{eqnarray*}

	Since $\frac{1}{kQ^k} - \frac{1}{(k+1)Q^{k+1}} >0$ for all $k > 0$, we can omit all but the first two terms to get
	\begin{eqnarray*}
		\ln\left(\frac{Q+1}{Q}\right)
		 & > & \frac{1}{Q} - \frac{1}{2Q^2} \\
		 & = & \frac{2Q - 1}{2Q^2}.
	\end{eqnarray*}

	Then 
	\begin{eqnarray*}
		\ln\left(\frac{Q+1}{Q}\right) - \frac{1}{Q+1}
		& > & \frac{2Q - 1}{2Q^2} - \frac{1}{Q+1} \\
		& = & \frac{(2Q-1)(Q+1) - 2Q^2}{2Q^2(Q+1)} \\
		& = & \frac{2Q^2 + Q - 1 - 2Q^2}{2Q^2(Q+1)} \\
		& = & \frac{Q-1}{2Q^2(Q+1)} \\
		& \geq & \frac{1}{Q^2(Q+1)} \\
		& > & \frac{1}{Q^4} \\
	\end{eqnarray*}
	where the last two inequalities use $Q \geq 3$.
\end{proof}

\end{document}

%% file: example_fishes.tex
\newcommand{\angledlabel}[3][]{\hspace{#2}\parbox{10ex}{{\bf #1}\newline\rotatebox{-15}{\textit{#3}}}}
\newcommand{\edgelab}[3]{\rotatebox{#1}{ \rotatebox{-#1}{\scriptsize #2} \rotatebox{-#1}{\scriptsize #3\phantom{xx}}}}

\begin{figure}[t]
  \centering\hspace{-8ex}
  \begin{tikzpicture}[xscale=1, yscale=1]
    \tikzstyle{sol}=[line width=1pt]
    \tikzstyle{nosol}=[dashed]

    \node[smallvertex] (root) at (0,0) {};

    \nextnodelab{r}{root}{-45:4.5}{revarc, nosol}{.4}{\edgelab{30}{0.0}{42}};
    \nextnodelab[leaf, label=below:{\angledlabel[F]{8ex}{Neo.~brev.}}]{rL}{r}{-45:2}{revarc, nosol}{.3}{\edgelab{30}{0.0}{1}};

    \nextnodelab{l}{root}{-135:4.5}{revarc, sol}{.4}{\edgelab{-30}{1.0}{36}}
    \nextnodelab[leaf, label=below:{\angledlabel[A]{8ex}{Neo.~simi.}}]{lL}{l}{-135:2}{revarc, sol}{0.4}{\edgelab{-30}{1.0}{10}}
    \nextnodelab[reti]{lr}{l}{-45:1}{revarc, sol}{0.5}{\edgelab{30}{0.4}{4}}
    \nextnodelab[leaf, label=below:{\angledlabel[B]{8ex}{Neo.~fasc.}}]{lrL}{lr}{-135:1}{revarc, sol}{0.5}{\edgelab{-30}{1.0}{6}}

    \nextnodelab{m}{root}{-90:1.44}{revarc, sol}{.4}{\edgelab{0}{0.784~}{3~~}}
    \nextnodelab{ml}{m}{-135:1.5}{revarc, sol}{.5}{\edgelab{-30}{0.64}{11}}
    \draw[arc, sol] (ml) -- (lr) node[pos=0.4] {\edgelab{-30}{0.4}{3}};
    \nextnodelab{mlr}{ml}{-45:1}{revarc, sol}{.5}{\edgelab{30}{0.4}{2}}
    \nextnodelab[leaf, label=below:{\angledlabel[C]{8ex}{Lam.~call.}}]{mlrL}{mlr}{-135:2}{revarc, nosol}{0.4}{\edgelab{-30}{0.0}{2}}
    \nextnodelab{mr}{m}{-45:2}{revarc, sol}{0.5}{\edgelab{30}{0.4}{47}}
    \nextnodelab[reti]{mrl}{mr}{-135:1.5}{revarc, sol}{0.5}{\edgelab{-30}{0.4}{56}}
    \draw[arc, sol] (mlr) -- (mrl) node[pos=0.5] {\edgelab{30}{0.4}{5}};
    \nextnodelab[reti]{mrr}{mr}{-45:1.5}{revarc, nosol}{0.5}{\edgelab{30}{0.0}{1}}
    \draw[arc, nosol] (r) -- (mrr) node[pos=0.3] {\edgelab{-30}{0.0}{4}};
    
    \nextnodelab[leaf, label=below:{\angledlabel[D]{8ex}{Neo.~wau.}}]{mrlL}{mrl}{-45:1}{revarc, sol}{0.4}{\edgelab{30}{1.0}{4}}
    \nextnodelab[leaf, label=below:{\angledlabel[E]{8ex}{Lam.~spec.}}]{mrrL}{mrr}{-45:1}{revarc, nosol}{0.5}{\edgelab{30}{0.0}{3}}

    \begin{pgfonlayer}{background}
      \path[fill=gray!25] (root.center) -- (r.center) -- (mrr.center) -- (mr.center) -- (mrl.center) -- (ml.center) -- (lr.center) -- (l.center) -- (root.center);
    \end{pgfonlayer}
	\end{tikzpicture}
	\caption{\looseness=-1
		A hypothesized heritage of several species of fish in a phylogenetic network~\cite{KDS+07}. 
		We take the inheritance probabilities to be~$0.4$ for reticulation edges and~1 for other edges.
        Edge weights are indicated by integers to the right of each edge.
		\emph{Edge weights and inheritance probabilities are not based on data and for illustrative purposes only.}
		The three reticulations are depicted as black filled vertices.
		The biggest subgraph without cut edges is shaded.
		The level and the reticulation number of the network are~3.
		It can be shown that the sets $\{\mbox{B},\mbox{D}\}$ 
		and $\{\mbox{B},\mbox{C},\mbox{D},\mbox{F}\}$ 
		maximize \NetPD{} among all size-2 and size-4 subsets of taxa, respectively.
        As an example, we illustrate how to compute the Network-PD score for~$Z=\{\mbox{A},\mbox{B},\mbox{D}\}$.
		The decimal numbers 
        left of the edges indicate the $\gam{Z}(e)$-values (see \Cref{def:gamma}), leading to a score of $\NetPD^p(Z)=195.968$.
        Dashed edges have $\gam{Z}(e)=0$ and hence do not contribute towards the Network-PD score.
	}
	\label{fig:example-network}
\end{figure}

%% file: psum_hard.tex

The reduction from \SubProd to \PenSum can be informally described as follows:
For an instance $(\{v_1,\dots v_m\}, M, k')$ of \SubProd and a big integer $A$, we let $a_i$ be (a rational close to) $A - \ln v_i$ and  let $b_i := 1/v_i$, for each $i \in [m]$. Let $Q := M$, let $k := k'$, and let $D$ be (a rational close to) $ kA - \ln Q - 1$.

Note that we cannot set $a_i := A - \ln v_i$ or $D := kA - \ln Q - 1$ exactly,
because in general these numbers are irrational and cannot be calculated exactly in finite time (nor stored in finite space).
Towards showing the correctness of the reduction, we temporarily forget about the need for rational numbers, and consider how the function $\sum_{i\in S}a_i - Q\cdot \prod_{i\in S}b_i$ behaves when we drop the `(a rational close to)' qualifiers from the descriptions above. In particular we will show that the function reaches its theoretical maximum exactly when $S$ is a solution to the \SubProd instance.

\subsubsection{Reduction with irrational numbers}

\begin{construction}\label{cons:irrational reduct}
  Let $(\{v_1,\dots v_m\}, M, k)$ be an instance of \SubProd.
  Let us define the following (not necessarily rational) numbers.
  \begin{itemize}
    \item Let $A := \lceil\max_{i \in [m]} (\ln v_i)\rceil+1$;
    \item Let $a_i^* := A - \ln v_i$ for each $i\in [m]$;
    \item Let $b_i := 1/v_i$ for each $i \in [m]$;
    \item Let $Q := M$;
    \item Let $D^*: = kA - \ln Q - 1$.
  \end{itemize}
  Finally, output the instance $(\{(a^*_i,b_i): i \in [m]\}, k, Q, D^*)$ of \PenSum.
\end{construction}
\medskip
We note that the purpose of $A$ in \Cref{cons:irrational reduct} is simply to ensure that $a_i^* > 0$ for each $i \in [m]$,
as required by the formulation of \PenSum.
Now, let $f^*:\binom{[m]}{k} \rightarrow \mathbb{R}$ be defined by
\begin{align*}
	f^*(S) := \sum_{i\in S}a^*_i - Q\cdot \prod_{i\in S}b_i.
\end{align*}

\begin{lemma}\label{lem:PenSumIrrationalMax}
  For any $S \in \binom{[m]}{k}$:
  \begin{enumerate}
    \item $f^*(S) \leq D^*$, 	and
    \item $f^*(S) = D^*$ if and only if $\prod_{i\in S} v_i = Q$.
  \end{enumerate}
\end{lemma}
\begin{proof}
  First, observe that given $|S| = k$, the function~$f^*$ can be written as 
	\begin{align*}
		f^*(S) = kA - \sum_{i \in S}\ln v_i - Q/\prod_{i\in S}v_i
		= kA - \ln \left(\prod_{i\in S}v_i\right) - Q/\prod_{i\in S}v_i\
	\end{align*}
	Letting $x_S := \prod_{i\in S}v_i$, we therefore have $f^*(S) = kA - \ln x_S - Qx_S^{-1}$.
	Let $g^*:\mathbb{R}_{>0} \to \mathbb{R}$ be defined by $g^*(x) := kA - \ln x - Qx^{-1}$
	and note that $f^*(S) = g^*(x_S)$ for any $S \in \binom{[m]}{k}$. 
	Recall that $g^*(x)$ has a critical point at $x'$ when $\frac{dg^*}{dx}(x') = 0$.
	Since $\frac{dg^*}{dx} = -x^{-1} + Qx^{-2}$, this occurs exactly when ${x'}^{-1} = Q{x'}^{-2}$, i.e. when $x' = Q$.
	Moreover, for $Q > x > 0$, we have $Qx^{-1} > 1$, implying
	\begin{align*}
	   \frac{dg^*}{dx} = -x^{-1} + Qx^{-2}  >  -x^{-1} + x^{-1} = 0.
	\end{align*}
	On the other hand, for $x > Q > 0$, we have $Qx^{-1} < 1$, implying
	\begin{align*}
	   \frac{dg^*}{dx} = -x^{-1} + Qx^{-2}  <  -x^{-1} + x^{-1} = 0.
	\end{align*}
	It follows that $g^*(x)$ is strictly increasing on the range $0 < x < Q$ and
    strictly decreasing on the range $x > Q$.
	Thus, $g^*(x)$ has a unique maximum on the range $x > 0$, and this maximum is achieved at $x = Q$.
	In particular, for all $S \in \binom{[m]}{k}$, we have	
	\begin{equation}
		f^*(S) = g^*(x_S) \leq g^*(Q) = kA - \ln Q - 1 = D^*.
	\end{equation}
	With equality if and only if $x_S = \prod_{i\in S}v_i = Q$.		
\end{proof}

The above result implies that, abusing terminology slightly, $(\{(a_i^*,b_i) \mid i \in [m]\}, k, Q, D^*)$ is a \yes-instance of `\PenSum' if and only if  $(\{v_1,\dots v_m\}, M, k')$ is a \yes-instance of \SubProd.

We are now ready to fully describe the polynomial-time reduction from \SubProd to \PenSum, showing how we can adapt the ideas above to work for rational $a_i$ and $D$.

\subsubsection{Reduction with rational numbers}

Let $(\{v_1,\dots v_m\}, M, k')$ be an instance of \SubProd, and
let $a_i^*, b_i, Q, k, D^*$ be defined as previously.
Then by \Cref{lem:PenSumIrrationalMax},
$f^*(S) = \sum_{i\in S}a^*_i - Q\cdot \prod_{i\in S}b_i \geq D^*$ if and only if $\prod_{i \in S}v_i = Q = M$
for any $S \in \binom{[m]}{k}$.

Our task now is to show how to replace $a_i^*$ and $D^*$ with rationals $a_i$ and $D$, in such a way that the same property holds (i.e. that $\sum_{i\in S} a_i - Q\cdot \prod_{i\in S} b_i \geq D$ if and only if $\prod_{i \in S} v_i = M$), and such that the instance  $(\{(a_i,b_i) \mid i \in [m]\}, k, Q, D)$ can be constructed in polynomial time.
The key idea is to find rational numbers that can be encoded in polynomially many bits but that are close enough to their respective irrationals that the difference between $f^*(S)$ and $f(S)$ (and between $D^*$ and $D$) is guaranteed to be small.
To this end, let us fix a positive integer $H$ to be defined later, and we will require all $a_i,b_i, D$ to be a multiple of $2^{-H}$. This ensures that the denominator part of any of these rationals can be encoded using $\Oh(H)$ bits.

\looseness=-1
Given any $x \in \mathbb{R}$ and a positive integer~$H$,
let $\floorH{x} : = r_x/2^H$, where $r_x$ is the largest integer such that~$r_x/2^H \leq x$.
For example $\floorvar{3}{\pi} = 3.125 = 25/2^3$, because $25/2^3 < \pi < 26/2^3$
(one may think of $\floorH{x}$ as the number derived from the binary representation of $x$ by deleting all digits more than $H$ positions after the binary point. Thus, as the binary expression of $\pi$ begins \texttt{11.00100 10000 11111}$\dots$, the binary expression of $\floorvar{3}{\pi}$ is \texttt{11.001}).
Similarly, let~$\ceilH{x}: = s_x/2^H$, where $s_x$ is the smallest integer such that $x \leq s_x/2^H$.
Finally, let~$\delta:= 1/{2^H}$.

\begin{observation}\label{obs:roundBounds}
	Let $x \in \mathbb{R}$.
    Then, $x- \delta   <   \floorH{x}  \leq x  \leq \ceilH{x} < x + \delta$.
\end{observation}

We can now describe the reduction from \SubProd to \PenSum.

\begin{construction}\label{cons:rational reduct}
  Let $(\{v_1,\dots v_m\}, M, k)$ be an instance of \SubProd.
  \begin{itemize}
    \item Let $A := \lceil\max_{i \in [m]} (\ln v_i)\rceil+1$;
    \item Let $a_i := \ceilH{a_i^*} = \ceilH{A - \ln v_i}$ for each $i\in [m]$;
    \item Let $b_i := 1/v_i$ for each $i \in [m]$;
    \item Let $Q := M$;
    \item Let $D := \floorH{D^*}  = \floorH{kA - \ln Q - 1}$.
  \end{itemize}
  Finally, output the instance $\Inst:=(\{(a_i,b_i) \mid i \in [m]\}, k, Q, D)$ of \PenSum.
\end{construction}
\medskip
In the following, we show that the two instances are equivalent.
To this end, let $f:\binom{[m]}{k} \rightarrow \mathbb{R}$ be defined by
\begin{align*}
	f(S) := \sum_{i\in S}a_i - Q\cdot \prod_{i\in S}b_i.
\end{align*}
Note that $f$ is the same as the function $f^*$ defined previously, but with each $a_i^*$ replaced by~$a_i$.
Then, $\Inst$ is a \yes-instance of \PenSum
if and only if
there is some $S \in \binom{[m]}{k}$ such that $f(S) \geq D$.
%
The next lemma shows the close relation between $f^*$ and $f$ (and between $D^*$ and $D$), which will be used in both directions to show the equivalence between \yes-instances of \SubProd and \PenSum.

\begin{lemma}\label{lem:fDbounds}
  Let $S \in \binom{[m]}{k}$. Then,
	$f^*(S) \leq f(S) < f^*(S) + k\delta$
	and
	$D^* - \delta < D \leq D^*.$
\end{lemma}
\begin{proof}
  Observe that $f(S) - f^*(S) = \sum_{i \in S}(a_i - a_i^*)$ and $|S|=k$.
  Then, by \Cref{obs:roundBounds}, we have $0 \leq a_i - a_i^* < \delta$ for all $i \in [m]$.
  Thus, $0 \leq f(S) - f^*(S) < k \delta$, from which the first claim follows.
	The second claim follows immediately from \Cref{obs:roundBounds} and the fact that $D = \floorH{D^*}$.
\end{proof}

\begin{corollary}\label{cor:SubProdYesToPenSumYes}
  Let $S \in \binom{[m]}{k}$ such that $\prod_{i \in S}v_i = Q$.
  Then, $f(S) \geq D$.
\end{corollary}
\begin{proof}
    $D \stackrel{\text{Lem.~\ref{lem:fDbounds}}}{\leq} D^* \stackrel{\text{Lem.~\ref{lem:PenSumIrrationalMax}~(2)}}{=} f^*(S)
	\stackrel{\text{Lem.~\ref{lem:fDbounds}}}{\leq} f(S)$.
\end{proof}

We now have that $(\{v_1,\dots v_m\}, M, k')$ being a \yes-instance of \SubProd implies~$\Inst$ being a \yes-instance of \PenSum.
To show the converse, we show for all $S \in \binom{[m]}{k}$ that $\prod_{i \in S}v_i = Q' \neq Q$ implies $f(S) < D$.
Since $f(S) < f^*(S) + k\delta$ and $D^* - \delta < D$, it is sufficient to show that $f^*(S) + k\delta \leq D^* - \delta$,
that is $(k+1)\delta \leq D^* - f^*(S)$.
To do this, we first establish a lower bound on $D^* - f^*(S')$ in terms of $Q$, using the following technical lemma, whose proof is deferred to the appendix.

\begin{lemma}\label{lem:logDifferenceBound}
	Let $Q,Q' \in \mathbb{N}_+$ with $Q \geq 2$ and $Q\neq Q'$.
    Then, $\ln Q' - \ln Q + Q/Q' - 1 > Q^{-4}$.
\end{lemma}
We explicitly note that we use the natural logarithm. For other logarithms, say $log_2$, this lemma is not true. For example for $Q=2$ and $Q'=1$ we have $\log_2(1) - \log_2(2) + 2/1 -1 = 0 - 1 + 2 - 1 = 0 < 2^{-4}$.

\begin{corollary}\label{cor:deltaGap}
	 Suppose $\prod_{i \in S}v_i = Q' \neq Q$ for some $Q\ge 2$ and $S \in \binom{[m]}{k}$. Then $D^* - f^*(S) > Q^{-4}$.
\end{corollary}
\begin{proof}
	Recall that $D^* = kA - \ln Q - 1$ and that $f^*(S) = kA - \ln(\prod_{i \in S}v_i) - Q/(\prod_{i \in S}v_i) = kA - \ln Q' - Q/Q'$.
	Then $D^* - f^*(S) = \ln Q' - \ln Q + Q/Q' - 1$. It follows from \Cref{lem:logDifferenceBound} that $D^* - f^*(S) > Q^{-4}$.
\end{proof}

Given the above we can now fix a suitable value for $H$. Given that we wanted $D^* - f^*(S) \geq (k+1)\delta = \nicefrac{(k+1)}{2^H}$ when $\prod_{i \in S} v_i \neq Q$, and assuming without loss of generality that $k < Q$, it is sufficient to set $H = 5\lceil\log_2 Q\rceil$.

\begin{corollary}\label{cor:deltaGap2}
	Let $H = 5 \lceil\log_2 Q\rceil$ and $\delta = (1/2^H)$.
	Then for $(\{(a_i, b_i) \mid i \in [m]\}, k, Q, D)$ constructed as above, it holds that $Q^{-4} \geq (k+1)\delta$.
\end{corollary}
\begin{proof}
	W.l.o.g.\ we may assume $k < Q$. 
	Then, $(k+1)\delta \leq Q/2^H \leq Q/Q^5 = Q^{-4}$.
\end{proof}
We now have all necessary pieces to reduce \SubProd to \PenSum.

\begin{theorem}
	\label{thm:NP-PenSum}
	\PenSum is \NP-hard.
\end{theorem}
\begin{proof}
	Given an instance  $(\{v_1,\dots v_m\}, M, k')$ of \SubProd, let $Q:= M$, $H := 5 \lceil\log_2 Q\rceil$, and $\delta:= (1/2^H)$.
	Construct $A, a_i, b_i, k, D$ as described above, that is:
	$A := \lceil\max_{i \in [m]} (\ln v_i)\rceil+1$;
	$a_i := \ceilH{a_i^*} = \ceilH{A - \ln v_i}$ for each $i\in [m]$;
	$b_i := 1/v_i$ for each $i \in [m]$;
	$k := k'$;
	$D := \floorH{D^*}  = \floorH{kA - \ln Q - 1}$. Let $(\{(a_i, b_i) \mid i \in [m]\}, k, Q, D)$ be the resulting instance of \PenSum.

	We first show that $(\{(a_i, b_i)  \mid i \in [m]\}, k, Q, D)$ is a \yes-instance of \PenSum if and only if $(\{v_1,\dots v_m\}, M, k')$ is a \yes-instance of \SubProd.
	Suppose first that $(\{v_1,\dots v_m\}, M, k')$ is a \yes-instance of \SubProd. Then there is some $S \in \binom{[m]}{k}$ such that $\prod_{i \in S} v_i = M = Q$.
	Then by \Cref{cor:SubProdYesToPenSumYes}, $f(S) \geq D$ and so $(\{(a_i, b_i)  \mid i \in [m]\}, k, Q, D)$ is a \yes-instance of \PenSum.
   
	Conversely, suppose that $(\{(a_i, b_i)  \mid i \in [m]\}, k, Q, D)$ is a \yes-instance of \PenSum.
	Then there is some $S \in \binom{[m]}{k}$ such that $f(S) \geq D$.
	By \Cref{lem:fDbounds} and \Cref{cor:deltaGap2}, we have that $f^*(S) > f(S) - k\delta \geq D - k\delta > D^* - (k+1)\delta \geq D^* - Q^{-4}$. Thus $D^* - f^*(S) \leq Q^{-4}$, which by \Cref{cor:deltaGap} implies that $\prod_{i \in S} v_i = Q$, and so $(\{v_1,\dots v_m\}, M, k')$ is a \yes-instance of \SubProd.

	It remains to show that the reduction takes polynomial time.
    For this, it is sufficient to show that the rationals $A,k,D$ and $a_i,b_i$ for $i \in [m]$ can all be calculated in polynomial time.
    Observe that  $A = \lceil\max_{i \in [m]} (\ln v_i)\rceil$ is the unique integer such that $e^A > \max_{i \in [m]} v_i > e^{A-1}$.
    Since $\ln v_i < \log_2 v_i$, we have
    $1\leq A \leq \lceil\max_{i \in [m]} \log_2 v_i\rceil$ and so we can find $A$ in polynomial time by checking all integers in this range.

    For each $i \in [m]$, $a_i = \ceilH{A - \ln v_i} = r_i/2^H$, where $r_i$ is the minimum integer such that $A - \ln v_i \leq r_i/2^H$.
    Thus, we can compute $r_i$ by checking $e^{A - r_i/2^H}\leq v_i$ with $r_i = 2^H\cdot (A - \ceilH{\ln v_i})$,
    setting $r_i$ to its successor if the inequality is not satisfied.
    %
    Thus we can construct $a_i$ in polynomial time, and $a_i$ can be represented in $\Oh(\log_2 r + H)$ bits.
    The construction of $D$ can be handled in a similar way.
 
	For each $i \in [m]$, rational $b_i = 1/v_i$ can be represented in $\Oh(\log_2 v_i)$ bits
    (recall that we represent $1/v_i$ with binary representations of the integers~$1$ and $v_i$)
    and takes $\Oh(\log_2 v_i)$ time to construct. 
	$Q$ and $k$ are taken directly from the instance $(\{v_1,\dots v_m\}, M, k')$.
\end{proof}

%% file: nd_hard.tex
Finally, reducing from \ucNAP, we show the following main result.

\begin{theorem}
	\label{thm:level-1}
	\MaxNPD is \NP-hard even if the input network has level~1, all weights are positive, and the distance between the root and each leaf is~4.
\end{theorem}

\begin{proof}
Because \PenSum is \NP-hard, we know that \ucNAP is \NP-hard on trees of height~2~\cite{komusiewicz2023multivariate}.
Let \Tree be an $L$-tree of height~2 for some $L$ and let an instance $\Instance = (\Tree, \w, q, k, D)$ of \ucNAP be given.

\begin{figure}[t]
	\centering
	\begin{tikzpicture}[scale=1,every node/.style={scale=0.8}]
		\tikzstyle{txt}=[circle,fill=white,draw=white,inner sep=0pt]
		\tikzstyle{nde}=[circle,fill=black,draw=black,inner sep=2.5pt]
		\node[txt] at (5.1,2.9) {$QM^2$};
		\node[txt] at (5,3.4) {$q/2$};
		\node[txt] at (6,3.9) {$q/(2-q)$};
		
		\node[smallvertex,inner sep=3.5pt, label=right:$\ell$] (l) at (5,4.5) {};
		\node[smallvertex,inner sep=3.5pt, label=left:$v_\ell^{1}$] (vl1) at (4,4) {};
		\node[smallvertex,inner sep=3.5pt, label=right:$v_\ell^{2}$] (vl2) at (5.5,3.5) {};
		\node[leaf,inner sep=4.5pt, label=left:$\ell^{-}$] (lm) at (4,3) {};
		\node[leaf,inner sep=4.5pt, label=right:$\ell^{*}$] (ls) at (5.5,2.5) {};
		
		\draw[dashed,-stealth] (5.4,5.2) to[bend right] (l);
		\draw[thick,-stealth] (l) to (vl1);
		\draw[thick,-stealth] (l) to (vl2);
		\draw[thick,-stealth] (vl1) to (vl2);
		\draw[thick,-stealth] (vl1) to (lm);
		\draw[thick,-stealth] (vl2) to (ls);
	\end{tikzpicture}
	\caption{Illustration of the leaf-gadget. Omitted edge-weights are~1 and $q(\ell)$ is abbreviated to $q$.}
	\label{fig:leaf-gadget}
\end{figure}%
We define a leaf-gadget which is illustrated in \Cref{fig:leaf-gadget}.
Let $\ell\in L$ be a leaf with success-probability $q(\ell)$.
Add four vertices $v_\ell^{1}, v_\ell^{2}, \ell^*, \ell^-$ and edges $\ell v_\ell^{1}, \ell v_\ell^{2}, v_\ell^{1}v_\ell^{2}, v_\ell^{1}\ell^-$, and~$v_\ell^{2}\ell^*$.
The only reticulation in this gadget is $v_\ell^{2}$ with incoming edges $\ell v_\ell^{2}$ and $v_\ell^{1}v_\ell^{2}$.
We set the inheritance probabilities $p(\ell v_\ell^{2}) := q(\ell) / (2 - q(\ell))$ and $p(v_\ell^{1}v_\ell^{2}) := q(\ell) / 2$ which are both in~$[0,1]$ because~$q(\ell) \in [0,1]$.

Let \Net be the network that results from replacing each leaf of \Tree with the corresponding leaf-gadget.
The leaves of \Net are $L' := \{ \ell^*, \ell^- \mid \ell \in L \}$.
Let $d$ denote the largest denominator in a success-probability  $q(\ell)$ of a leaf $\ell$ of $\Tree$, so that every $q(\ell)$ is expressible as $c'/d'$ for some pair of integers $c,d$ such that $d' \leq d$. 
Let $M$ and $Q$ be large integers, such that $M$ is bigger than $\PDsub{\Tree}(L) \ge |L| \ge k$, 
and $Q\cdot D$ and $Q\cdot d^{-k}$ are both bigger than~3.

Observe that the number of bits necessary to write $M$ and $Q$ is polynomial in the size of~\Instance.
We set the weight of edges $e \in E(\Tree)$ in \Net to $\w'(e) = kQ \cdot \w(e)$.
For each $\ell\in L$ we set~$\w'(v_\ell^{2}\ell^*) := Q\cdot M^2$ and $\w'(e) := 1$ for $e \in \{\ell v_\ell^{1}, \ell v_\ell^{2}, v_\ell^{1}v_\ell^{2}, v_\ell^{1}\ell^-\}$.

Finally let $\Instance' := (\Net, \w', p, k, D' := kQ(M^2 + D))$ be an instance of \MaxNPD.
Each leaf-gadget is a level-1 network.
As the leaf-gadgets are connected by a tree, \Net is a level-1 network.
Recall that the height of the tree \Tree is 2, and as such the distance between the root and each leaf in in \Net  is~4.


Before showing that  \Instance and $\Instance'$ are equivalent, we show that $\gam{Z}(e) = q(\ell)$ in the case that $\ell^*\in Z$ but $\ell^-\not\in Z$. Indeed because $\ell^-\not\in Z$, we conclude that $\gam{Z}(\ell v_\ell^{2}) = p(\ell v_\ell^{2}) = q(\ell) / (2 - q(\ell))$ and $\gam{Z}(\ell v_\ell^{1}) = \gam{Z}(v_\ell^{1}v_\ell^{2}) = p(v_\ell^{1}v_\ell^{2}) = q(\ell) / 2$.
Subsequently, 
\begin{equation}\label{eqn:gamma}
	\gam{Z}(e)
	= 1 - (1 - \gam{Z}(\ell v_\ell^{1}))(1 - \gam{Z}(\ell v_\ell^{2})) 
	= 1 - \frac{2 - q(\ell)}{2} \cdot \frac{2 - 2q(\ell)}{2 - q(\ell)}
	= 1 - \frac{2 - 2q(\ell)}{2} 
	= q(\ell).
\end{equation}

``$\Rightarrow$'': Suppose that \Instance is a \yes-instance of \ucNAP and that $S\subseteq L$ is a solution of \Instance, that is~$|S|\le k$ and~$\PDsub{\Tree}(S) \ge D$.
Let $S' := \{ \ell^* \mid \ell \in S \}$ be a subset of~$L'$.
Clearly $|S'| = |S| \le k$.
Because~\Tree does not contain reticulation edges and $\gam{Z}(e) = q(\ell)$ with $e$ being the edge incoming at $\ell$, we conclude that
\begin{align*}
	\NetPD(S')
	\ge kQ\cdot \PDsub{\Tree}(S) + k\cdot \w'(v_\ell^{2}\ell^*)
	\ge kQ\cdot (D + M^2)
	= D'
\end{align*}
hence, $S'$ is a solution of $\Instance'$.

``$\Leftarrow$'':
Let $S'$ be a solution of $\Instance'$.
Let $S^- = S \cap \{ \ell^- \mid \ell\in L \}$ and $S^* = S \cap \{ \ell^* \mid \ell\in L \}$.
Towards a contradiction, assume 
$S^- \neq \emptyset$.
Then, however, using $3 < Q\cdot D$,
\begin{align*}
	\NetPD(S')
	& \le  \sum_{\ell^- \in S^-}(\w'(v_\ell^{1}\ell^-) + \w'(\ell v_\ell^{1})) +  |S^*|(QM^2 + 3)
     + \sum_{e\in E(\Tree)} \w'(e)\\
 & \le 2|S^-| + |S^*|(QM^2 + 3) +  kQM \\
	& \le  2 + (k-1)(QM^2 + 3) + kQM \\
	&	< kQM^2 - QM^2 + kQM + 3k\\
	&   < k(QM^2 + 3)
		< k(QM^2 + QD)
		= D'
\end{align*}
contradicts that $S'$ is a solution.
Therefore, we conclude that $S' \subseteq \{ \ell^* \mid \ell\in L \}$ and $|S'|=k$.
Define $S := \{ \ell \mid \ell^*\in S' \}$.
Subsequently, with \eqref{eqn:gamma} we conclude 
\begin{align*}
	kQ(M^2 + D) 
  = D' 
  & \le \NetPD(S') \\
  & = k\cdot QM^2 + \sum_{\ell \in S}\bigg(\underbrace{\frac{q(\ell)}2 + \frac{q(\ell)}2 + \frac{q(\ell)}{2-q(\ell)}}_{\leq 3}\bigg) + kQ\cdot \PDsub{\Tree}(S).
\end{align*}
It follows that $\PDsub{\Tree}(S) \ge \nicefrac{1}{kQ} \cdot (kQ(M + D) - kQM - 3k) = D - 3/Q$.

It remains to show that $\PDsub{\Tree}(S)$ cannot take any values in the range $[D - 3/Q, D)$, i.e. that $\PDsub{\Tree}(S) \ge D - 3/Q$ implies that  $\PDsub{\Tree}(S) \ge D$.
To this end, let $c_\ell, d_\ell$ be the unique positive integers such that $q(\ell) = c_\ell/d_\ell$ for each leaf $\ell$ in $\Tree$.
Then $q(\ell)$ is a multiple of $1/d_\ell$ by construction, as is $(1 - q(\ell))$.
It follows that for any edge $e$ in $\Tree$, $\gamma'_S(e) = (1 - \prod_{\ell \in \off(e) \cap S}(1- q(\ell)))$ is a multiple of $1/(\prod_{\ell \in S}d_\ell)$.
As all edge weights are integers, we also have that $\PDsub{\Tree}(S)$ is a multiple of $1/(\prod_{\ell \in S}d_\ell)$. It follows that either $\PDsub{\Tree}(S) \geq D$ or $D - \PDsub{\Tree}(S) \geq 1/(\prod_{\ell \in S}d_\ell)$.
As $d_\ell \leq d$ for any $\ell$, this difference is at least $d^{-k} > 3/Q$.
It follows that if $\PDsub{\Tree}(S) \ge D - 3/Q$ then in fact $\PDsub{\Tree}(S) \ge D$.

We conclude $\PDsub{\Tree}(S)\ge D$.
Hence with $|S| = |S'| \le k$ we conclude that $S$ is a solution of~\Instance. Thus, \Instance is a \yes-instance of \ucNAP.
\end{proof}

%% file: main-ArXiv.bbl
\begin{thebibliography}{20}
	\providecommand{\natexlab}[1]{#1}
	\providecommand{\url}[1]{\texttt{#1}}
	\expandafter\ifx\csname urlstyle\endcsname\relax
	\providecommand{\doi}[1]{doi: #1}\else
	\providecommand{\doi}{doi: \begingroup \urlstyle{rm}\Url}\fi
	
	\bibitem[{Atkin} and {Bernstein}(2004)]{AtkinBernstein04}
	A.~O.~L. {Atkin} and D.~J. {Bernstein}.
	\newblock {Prime sieves using binary quadratic forms}.
	\newblock \emph{Mathematics of Computation}, 73:\penalty0 1023--1030, 2004.
	
	\bibitem[Bordewich et~al.(2022)Bordewich, Semple, and
	Wicke]{bordewich2022complexity}
	Magnus Bordewich, Charles Semple, and Kristina Wicke.
	\newblock On the {C}omplexity of optimising variants of {P}hylogenetic
	{D}iversity on {P}hylogenetic {N}etworks.
	\newblock \emph{Theoretical Computer Science}, 917:\penalty0 66--80, 2022.
	
	\bibitem[Dusart(1999)]{Dusart99}
	Pierre Dusart.
	\newblock {The $k$th Prime is Greater than $k(\ln k + \ln \ln k - 1)$ for $k
		\geq 2$}.
	\newblock \emph{Mathematics of Computation}, 68\penalty0 (225):\penalty0
	411--415, 1999.
	
	\bibitem[Faith(1992)]{faith1992conservation}
	Daniel~P. Faith.
	\newblock Conservation evaluation and phylogenetic diversity.
	\newblock \emph{Biological Conservation}, 61\penalty0 (1):\penalty0 1--10,
	1992.
	
	\bibitem[Fellows et~al.(2012)Fellows, Gaspers, and Rosamond]{numnum}
	Michael~R. Fellows, Serge Gaspers, and Frances~A. Rosamond.
	\newblock Parameterizing by the number of numbers.
	\newblock \emph{Theory of Computing Systems}, 50:\penalty0 675--693, 2012.
	
	\bibitem[Garey and Johnson(1979)]{GareyJohnson79}
	M.~R. Garey and David~S. Johnson.
	\newblock \emph{{Computers and Intractability: {A} Guide to the Theory of
			NP-Completeness}}.
	\newblock W. H. Freeman, 1979.
	\newblock ISBN 0-7167-1044-7.
	
	\bibitem[Huson et~al.(2010)Huson, Rupp, and Scornavacca]{huson2010phylogenetic}
	Daniel~H. Huson, Regula Rupp, and Celine Scornavacca.
	\newblock \emph{{Phylogenetic Networks: Concepts, Algorithms and
			Applications}}.
	\newblock Cambridge University Press, 2010.
	
	\bibitem[Jones and Schestag(2023)]{MAPPD}
	Mark Jones and Jannik Schestag.
	\newblock {How Can We Maximize Phylogenetic Diversity? Parameterized Approaches
		for Networks}.
	\newblock In \emph{Proceedings of the 18th International Symposium on
		Parameterized and Exact Computation (IPEC 2023)}, pages 30:1--30:12.
	Schloss-Dagstuhl-Leibniz Zentrum f{\"u}r Informatik, 2023.
	
	\bibitem[Koblm{\"u}ller et~al.(2007)Koblm{\"u}ller, Duftner, Sefc, Aibara,
	Stipacek, Blanc, Egger, and Sturmbauer]{KDS+07}
	Stephan Koblm{\"u}ller, Nina Duftner, Kristina~M. Sefc, Mitsuto Aibara, Martina
	Stipacek, Michel Blanc, Bernd Egger, and Christian Sturmbauer.
	\newblock {Reticulate phylogeny of gastropod-shell-breeding cichlids from Lake
		Tanganyika--the result of repeated introgressive hybridization}.
	\newblock \emph{BMC Evolutionary Biology}, 7:\penalty0 1--13, 2007.
	
	\bibitem[Komusiewicz and Schestag(2025)]{komusiewicz2023multivariate}
	Christian Komusiewicz and Jannik Schestag.
	\newblock {A Multivariate Complexity Analysis of the Generalized Noah's Ark
		Problem}.
	\newblock Elsevier, 2025.
	
	\bibitem[Merz et~al.(2023)Merz, Barnard, Rees, Smith, Maroni, Rhodes, Dederer,
	Bajaj, Joy, Wiedmann, and Sutherland]{doi:10.1177/00368504231201372}
	Joseph~J. Merz, Phoebe Barnard, William~E. Rees, Dane Smith, Mat Maroni,
	Christopher~J. Rhodes, Julia~H. Dederer, Nandita Bajaj, Michael~K. Joy,
	Thomas Wiedmann, and Rory Sutherland.
	\newblock {World scientists’ warning: The behavioural crisis driving
		ecological overshoot}.
	\newblock \emph{Science Progress}, 106\penalty0 (3):\penalty0
	00368504231201372, 2023.
	
	\bibitem[Minh et~al.(2006)Minh, Klaere, and von Haeseler]{PDinSeconds}
	Bui~Quang Minh, Steffen Klaere, and Arndt von Haeseler.
	\newblock {Phylogenetic Diversity within Seconds}.
	\newblock \emph{Systematic Biology}, 55\penalty0 (5):\penalty0 769--773, 10
	2006.
	
	\bibitem[Moret(1997)]{Moret97}
	Bernard~M. Moret.
	\newblock \emph{The theory of computation}.
	\newblock Addison-Wesley, 1997.
	
	\bibitem[Pardi and Goldman(2005)]{pardi2005species}
	Fabio Pardi and Nick Goldman.
	\newblock Species choice for comparative genomics: being greedy works.
	\newblock \emph{PLoS Genetics}, 1\penalty0 (6):\penalty0 e71, 2005.
	
	\bibitem[Ripple et~al.(2017)Ripple, Wolf, Newsome, Galetti, Alamgir, Crist,
	Mahmoud, Laurance, and 15]{10.1093/biosci/bix125}
	William~J. Ripple, Christopher Wolf, Thomas~M. Newsome, Mauro Galetti, Mohammed
	Alamgir, Eileen Crist, Mahmoud~I. Mahmoud, William~F. Laurance, and 364
	scientist signatories from 184~countries 15.
	\newblock {World Scientists’ Warning to Humanity: A Second Notice}.
	\newblock \emph{BioScience}, 67\penalty0 (12):\penalty0 1026--1028, 11 2017.
	\newblock ISSN 0006-3568.
	
	\bibitem[Rosser(1941)]{Rosser41}
	Barkley Rosser.
	\newblock {Explicit Bounds for Some Functions of Prime Numbers}.
	\newblock \emph{American Journal of Mathematics}, 63\penalty0 (1):\penalty0
	211--232, 1941.
	
	\bibitem[Steel(2005)]{steel2005phylogenetic}
	Mike Steel.
	\newblock Phylogenetic diversity and the greedy algorithm.
	\newblock \emph{Systematic Biology}, 54\penalty0 (4):\penalty0 527--529, 2005.
	
	\bibitem[Weitzman(1998)]{weitzman1998noah}
	Martin~L. Weitzman.
	\newblock The {Noah}'s ark problem.
	\newblock \emph{Econometrica}, pages 1279--1298, 1998.
	
	\bibitem[Wicke and Fischer(2018)]{WickeFischer2018}
	Kristina Wicke and Mareike Fischer.
	\newblock Phylogenetic diversity and biodiversity indices on phylogenetic
	networks.
	\newblock \emph{Mathematical Biosciences}, 298:\penalty0 80--90, 2018.
	
	\bibitem[Wicke et~al.(2021)Wicke, Mooers, and Steel]{wicke2021formal}
	Kristina Wicke, Arne Mooers, and Mike Steel.
	\newblock Formal links between feature diversity and phylogenetic diversity.
	\newblock \emph{Systematic Biology}, 70\penalty0 (3):\penalty0 480--490, 2021.
	
\end{thebibliography}
